\pdfoutput=1
\documentclass[10pt,twocolumn]{article}
\usepackage[margin=1in]{geometry}
\usepackage{amsmath,amssymb,amsthm,bm,mathtools}
\usepackage{graphicx}
\usepackage{booktabs}
\usepackage{microtype}
\usepackage[numbers,sort&compress]{natbib}
\usepackage[colorlinks,linkcolor=blue!60!black,citecolor=blue!60!black,urlcolor=blue!60!black]{hyperref}
\usepackage{xcolor}
\newtheorem{proposition}{Proposition}
\newtheorem{lemma}{Lemma}
\newcommand{\Tr}{\mathrm{Tr}}
\newcommand{\Z}{\mathbb Z}
\newcommand{\GSD}{\mathrm{GSD}}
\newcommand{\Had}{\mathsf H}

\title{\bfseries\fontsize{14.4}{18}\selectfont An exactly solvable obstructed seam on a lattice Klein bottle:\\ bulk-locked Majorana rings and register-dependent toric codes}
\author{Junhao (Eigen) Wang\\[2pt] \small Department of Mathematics, New York University, New York, NY, USA\\ \small \texttt{jw10111@nyu.edu}}
\date{\small Preprint --- 19 October 2026}

\begin{document}
\maketitle

\begin{abstract}
\noindent
We study a frustrated seam produced by literal reflection gluing of Wen's plaquette model on a closed lattice Klein bottle. The seam operators anticommute along two cycles for even circumference and one cycle of twice the length for odd circumference. Their algebra reduces to free Majorana rings, while the commuting bulk stabilizers constrain the allowed ring sectors. This bulk locking yields exact ground energies, excitation gaps and degeneracies, checked by phase-aware Pauli traces and finite-size diagonalization, including two 24-qubit energy predictions. The seam degeneracy is not topological: the single-site seam fields tested split it at second order. We contrast this obstructed gluing with consistent toric-code realizations. An edge register preserves the anyon labels; a Hadamard-dressed site register has colour-dependent seam transport, with ground-state degeneracy 2 or 4 for even circumference and a sign-frustrated exception at odd circumference. For the enumerated lattice isometries and sizes, we compute logical gates and distances, distinguishing abstract qubit relabelling from local hardware depth. Finally, on three 16- and 24-qubit lattices, chiral matter selects a holonomy doublet whose electric-field splitting follows two competing tunnelling routes. These matter results remain finite-size statements. The main result is the exact obstructed-seam solution with bulk sector locking; the register, gate and matter calculations delimit its relation to topological codes.
\end{abstract}

\section{Introduction}\label{sec:intro}

We study the local Hamiltonian obtained when the reflected corners of Wen's plaquette model are identified literally on a closed lattice Klein bottle. The identification is incompatible with the plaquette symmetry, but it defines an exactly solvable frustrated seam. Our central result is the reduction of this seam to Majorana rings whose allowed sectors are locked by the commuting bulk stabilizers. The consistent edge and site registers provide the comparison that isolates this obstruction; their logical gates and a finite matter-coupled example are developed as related consequences.

Twist defects of the toric code---the end points of lines across which the electric and magnetic anyons $e$ and $m$ are exchanged---bind Majorana modes and realize Ising-type non-Abelian statistics \cite{Bombin2010,YouWen2012,BarkeshliJianQi2013,TeoRoyChen2014}. On an orientable lattice such lines are put in by hand, by dislocations or by an explicit change of the stabilizers along a cut. A non-orientable surface offers a second route: a reflection gluing may or may not act on the anyon labels, and when it does, the seam is a duality twist line that closes on itself without end points. The general theory of topological order on non-orientable surfaces relates the ground-state degeneracy to the anyons left invariant by the parity action \cite{ChanTeoRyu2016}, and cross-cap defects have recently been shown to implement fault-tolerant logical gates in the surface code \cite{Kobayashi2024crosscap}. Lattice Klein bottles have at the same time become a common object in the theory of symmetry-protected and gauge-theoretic phases: the momentum-space Brillouin Klein bottle of $\pi$-flux lattices \cite{ChenYangZhao2022,XiaoZhao2026}, the real-space Klein-bottle quotients of staggered fermions with their magnetic translation algebra \cite{SeibergZhang2026}, and symmetry defects on Klein bottles as probes of SPT responses \cite{RavindranYangChen2026}.

We ask whether identifying the top and bottom rows of an $L_x\times L_y$ periodic $\Z_2$ toric-code strip by $x\mapsto a-x$ produces a Klein-bottle seam that exchanges $e$ and $m$. We find that the question has no answer at the level of the topological order alone: it is decided by the register---by whether the code is written with qubits on the edges (Kitaev's form \cite{Kitaev2003}, stars and plaquettes) or with qubits on the sites (Wen's plaquette model \cite{Wen2003plaquette}, $F_p=X_{SW}Z_{SE}X_{NE}Z_{NW}$). The two registers describe the same topological order and are related by a local unitary, but lattice reflections act on them differently: in the edge register a reflection maps stars to stars and plaquettes to plaquettes; in the site register a reflection reverses the cyclic $X,Z,X,Z$ pattern of every plaquette and is a symmetry only when composed with the global Hadamard $\Had^{\otimes N}$, which restores the corner pattern; colour transport separately decides whether the seam exchanges $e$ and $m$ (Sec.~\ref{sec:register}). This is the lattice form of a symmetry-enrichment statement: the register fixes the permutation action of the reflection on anyons, and the Klein bottle inherits it.

The exact obstructed-seam solution is the central result, item (iii) below; the register comparison establishes its setting, and items (iv)--(v) give code and finite-matter consequences. (i) In the edge register every reflection gluing closes into a stabilizer code with $\GSD=4$, for every axis offset and either parity of $L_y$ (Sec.~\ref{sec:register}). (ii) In the site register the Hadamard-dressed gluing closes, for even $L_x$, with $\GSD=2$ or $4$, decided by a colour-holonomy rule---whether the plaquette two-colouring that labels $e$ and $m$ is broken across the seam---which also reproduces the known degeneracy of the model on odd tori and on the sheared torus and extends to a fixed-point-free glide quotient (for odd $L_x$ the dressed plaquettes commute but multiply to $-1$: the model is frustrated by a global sign, Sec.~\ref{sec:register}); the values $4$ and $2$ are the numbers of parity-invariant anyons for a trivial and for an $e\leftrightarrow m$ parity action \cite{ChanTeoRyu2016}. (iii) The literal, undressed Pauli gluing in the site register never closes, and the local Hamiltonian it defines is exactly solvable (Sec.~\ref{sec:seam}): the $2L_x$ plaquette operators adjacent to the seam anticommute along two cycles of length $L_x$ (even $L_x$) or one cycle of length $2L_x$ (odd $L_x$), each cycle is the algebra of a free Majorana ring, the bulk stabilizers lock the rings to the antiperiodic sector, and the ground energy, the gap, the low spectrum and the ground-state degeneracy follow in closed form. Two predictions at 24 qubits, $E_0=-20$ for the mirror gluing and $E_0=-12-4\sqrt3$ for the glide quotient with odd shift, are confirmed to $10^{-13}$. The seam degeneracy is exact for the frustrated model but not topological: single-site fields on the seam rows lift it at second order. (iv) The consistent gluings are Klein-bottle codes (Sec.~\ref{sec:codes}), and the audited lattice-isometry family implements logical gates: the half-period translation along the seam, which on a Klein bottle is not isotopic to the identity, acts on the two logical qubits as the Dehn twist about the two-sided cycle---a logical $\mathrm{SWAP}$ in the basis in which each logical qubit lives on one cross-cap---and, in the site register only, the colour-exchanging isometries act as $\Had\otimes\Had$. Qubit relabellings and single-row Hadamards thus generate $\Z_2\times\Z_2$, the group generated by the mapping-class action on $H_1(K;\Z_2)$ and by $e\leftrightarrow m$; the edge register gets only the $\mathrm{SWAP}$, and the twisted code gets only logical Paulis from the audited family. We compute all logical operators and distances exactly, up to $48$ qubits, for even $L_x$. (v) Coupling the gauge field to the chiral matter of the companion paper~\cite{Wang2026paper1} (Sec.~\ref{sec:matter}) on three small lattices, the matter splits the four holonomy sectors $2+2$ and selects the antiperiodic pair, a two-fold degenerate doublet with $\bar Z$ a seam-crossing Wilson loop and $\bar X$ the half translation. Its splitting under an electric field is of order $h^{\min(L_x,2L_y)}$, because the one-sided cycle, although shorter, must be traversed twice to return to the selected pair; the leading-order coefficient is computed exactly and matches exact diagonalization to better than one per cent, and the powers $4$ and $6$ are seen on $16$- and $24$-qubit lattices. These are finite-size statements, and we do not extrapolate them.

We are explicit about what is and is not new. That suitable lattice symmetries can act as $e\leftrightarrow m$ in Wen's model, that such lines carry Majorana degrees of freedom, and that the Klein-bottle degeneracy counts parity-invariant anyons are known \cite{Wen2003plaquette,Bombin2010,YouWen2012,TeoRoyChen2014,ChanTeoRyu2016}. That the toric code on a surface with cross-caps admits a fault-tolerant Hadamard from $e\leftrightarrow m$ is the result of Ref.~\cite{Kobayashi2024crosscap}, and that Dehn twists can be implemented by qubit permutations with constant-depth corrections is known \cite{Zhu2020dehn,Lavasani2019}. What we add is the register comparison on the same lattice Klein bottle with exact stabilizer counts, the single colour-holonomy rule for the dressed gluings, the anatomy of the obstructed gluing---a closed, end-point-free seam whose algebra is one or two locked Majorana rings according to the parity of $L_x$, with exact formulas verified on the lattice---and the table of logical gates implemented by the enumerated deck-compatible affine isometries on the three Klein-bottle codes at the stated sizes, in which the Dehn twist is a rigid translation and the register decides whether the cross-cap Hadamard is available at all. The model that motivates this work---the same $\pi$-flux lattice with $p$-wave pairing, which on the same Klein bottle is forced to form a chirality wall along the seam---lives in the site register; Sec.~\ref{sec:matter} couples it to the gauge field on the lattices we can reach exactly, and Sec.~\ref{sec:discussion} says what is and is not established by that.

\section{Two registers and three gluings}\label{sec:model}

\begin{figure*}[!t]
\centering\includegraphics[width=0.98\textwidth]{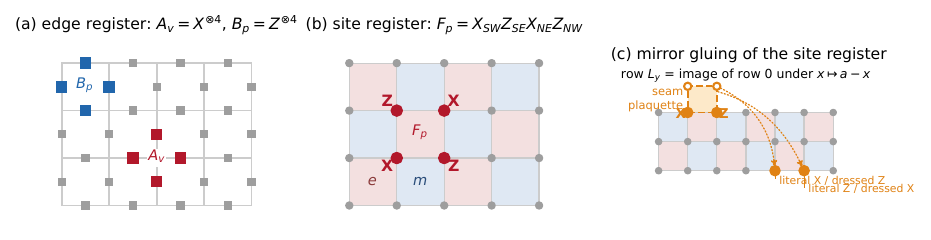}
\caption{The two registers and the mirror seam. (a) Edge register: qubits on the edges of the square lattice, stabilizers $A_v=X^{\otimes4}$ on the four edges at a vertex and $B_p=Z^{\otimes4}$ around a plaquette. (b) Site register: qubits on the sites, one operator $F_p=X_{SW}Z_{SE}X_{NE}Z_{NW}$ per plaquette; the plaquettes are two-coloured by $(x+y)\bmod2$ and the two colours carry the two anyon types. (c) Mirror gluing of the site register: the row $y=L_y$ is the image of row $0$ under $x\mapsto a-x$, so the two top corners of a plaquette of row $L_y-1$ land on the sites $a-x-1$ and $a-x$ of row $0$. The literal gluing keeps $X$ on the image of the NE corner and $Z$ on the image of the NW corner and is obstructed; the dressed gluing swaps them (the identification includes the Hadamard) and closes.}
\label{fig:registers}
\end{figure*}

\paragraph{Lattices and registers.} We use an $L_x\times L_y$ square lattice with sites $(x,y)$, $x\in\Z_{L_x}$ periodic and $y=0,\dots,L_y-1$, closed along $y$ by one of the gluings below. In the \emph{edge register} a qubit sits on every edge; the stabilizers are the vertex stars $A_v$ ($X$ on the four edges at $v$) and the plaquettes $B_p$ ($Z$ on the four edges of $p$), $H=-\sum_vA_v-\sum_pB_p$ \cite{Kitaev2003}. In the \emph{site register} a qubit sits on every site and there is one operator per plaquette,
\begin{equation}
F_p=X_{SW}\,Z_{SE}\,X_{NE}\,Z_{NW},\qquad H=-\sum_pF_p ,
\label{eq:wen}
\end{equation}
with the corners of $p=(x,y)$ at $(x,y),(x+1,y),(x+1,y+1),(x,y+1)$ \cite{Wen2003plaquette}. On the torus all $F_p$ commute; each site belongs to four plaquettes, and the plaquettes are two-coloured by $(x+y)\bmod2$. A violated plaquette of one colour is the $e$ anyon and of the other colour the $m$ anyon; a one-site translation exchanges the colours, which is the familiar statement that translation acts as $e\leftrightarrow m$ in this model \cite{Wen2003plaquette,TeoRoyChen2014}. The two registers describe the same bulk topological order; Definition~\ref{def:register} specifies the microscopic symmetry lifts as part of the register data. An identification of the bulk orders alone does not identify those lifts.

\paragraph{Gluings.} Row $L_y$ is identified with row $0$ through a map $s:\Z_{L_x}\to\Z_{L_x}$ of the column index: the torus, $s(x)=x$; the \emph{sheared torus}, $s(x)=x+L_x/2$ (orientation preserving); the \emph{mirror Klein bottle}, $s(x)=a-x$ with an axis offset $a\in\Z_{L_x}$ (site-centred for $a$ even, bond-centred for $a$ odd; orientation reversing); and, as a fixed-point-free alternative, the \emph{glide quotient} (even $L_x$ only): the torus $L_x\times2L_y$ modulo $g:(x,y)\mapsto(x+h,-y)$, $h=L_x/2$, whose fundamental domain consists of the full rows $y=1,\dots,L_y-1$ and the two half rows $y=0$ and $y=L_y$ (in which $(x,y)\sim(x+h,y)$), with $n=L_xL_y$ qubits like the others. In the edge register the wrap is applied to the vertical edges leaving row $L_y-1$ and to the top edges of the plaquettes of that row; the stabilizers are defined as qubit sets and no sign ambiguity arises. In the site register the plaquettes of row $L_y-1$ have their NE and NW corners at the sites $s(x+1)$ and $s(x)$ of row $0$ (Fig.~\ref{fig:registers}c). For the orientation-preserving gluings the cyclic order of the corners is preserved and $F_p$ is defined as in Eq.~\eqref{eq:wen}. For a reflection the cyclic order is reversed on the far side, and two definitions present themselves: the \emph{literal} gluing, which keeps $X$ on the image of the NE corner and $Z$ on the image of the NW corner---the $X,Z,X,Z$ pattern as read on the quotient plaquette---and the \emph{dressed} gluing, which conjugates the two far corners by the Hadamard, $X\leftrightarrow Z$. The dressed gluing is the quotient of the parent torus by the operation $G_1=\Had^{\otimes N}\circ(\text{reflection})$; we show in Sec.~\ref{sec:register} that $G_1$, and not the bare reflection, is the symmetry of the parent model, so that the dressed gluing is the symmetry quotient and the literal gluing does not implement a symmetry lift of these plaquette checks.

\paragraph{Methods.} Consistency, rank and degeneracy of a set of Pauli operators are exact linear algebra over $\mathrm{GF}(2)$: with each operator a row $(x|z)\in\mathrm{GF}(2)^{2n}$, two operators anticommute iff the symplectic product $x\cdot z'+z\cdot x'$ is odd, and a mutually commuting set $S$ of rank $r$ has $\GSD=2^{n-r}$ provided the group it generates does not contain $-1$, i.e.\ every $\mathrm{GF}(2)$ relation among the rows holds with phase $+1$; we compute the phase of every relation (a basis of the left null space of the row matrix, multiplied out as Pauli strings with phases), because the rank alone is blind to a relation of phase $-1$, in which case no state satisfies all operators and the Hamiltonian is frustrated. For an inconsistent set we record the number of anticommuting pairs, the graph they form (the anticommutation graph, vertices = operators) and the $\mathrm{GF}(2)$ rank of its adjacency matrix $\mathsf A$. Operator identities are evaluated with Pauli strings carrying their phases. Spectra come from exact diagonalization of the $2^n$-dimensional Hamiltonian (dense for $n\le13$, Lanczos otherwise; the Hamiltonians contain only $X$ and $Z$ and are real). Every number quoted below is produced by a script of the reproduction package (Appendix~\ref{app:code}).

\theoremstyle{definition}
\newtheorem{definition}{Definition}
\theoremstyle{remark}
\newtheorem{remark}{Remark}
\theoremstyle{plain}

\begin{definition}[Register and a lift of the gluing]\label{def:register}
Fix an Abelian topological order with fusion-label group $\mathcal A$ and braiding pairing, and fix the geometric identification $s$ on the covering lattice. A \emph{register realization} specifies a local tensor-product Hilbert space, its bulk check algebra, an identification of string-endpoint sectors with $\mathcal A$, and the microscopic lifts of the lattice identifications under consideration. An admissible lift $\widetilde s_R$ acts on local operators by the geometric map together with any specified internal dressing, preserves the bulk check constraints and their vacuum sector, and induces an automorphism $\rho_R(s)$ of the fusion labels. The \emph{seam action} is $P=\rho_R(s)$: a string of label $b$ continued across the seam has label $Pb$. The labels on both sides and the direction of crossing are fixed; for an orientation-reversing identification the braiding pairing is conjugated. Preservation of the check constraints need not mean preservation of a particular choice of Hamiltonian coefficients. The lift acts on the covering lattice and need not square to the identity: its square can include an internal action, such as charge conjugation $C:b\mapsto-b$.
\end{definition}

Thus ``edge'' and ``site'' name microscopic realizations of a symmetry action, not merely locations of qubits. The lift, label convention, offset and any translation in the gluing must be specified. For the edge code the reflection preserves the star and plaquette sectors. For the site code the Hadamard dressing restores the reversed corner pattern; the geometric colour transport then decides whether the continued label is $e$ or $m$. An attempted literal identification that does not preserve the local checks is not an admissible lift in Definition~\ref{def:register}. In this definition and Proposition~\ref{prop:register}, $P$ acts on anyon labels; it is distinct from the anticommutation adjacency matrix denoted by $\mathsf A$ in the methods and in Sec.~\ref{sec:seam}.

\section{Consistency and degeneracy: the role of the register}\label{sec:register}

\begin{proposition}[A register lift determines the marked seam transport]\label{prop:register}
Consider two register realizations $R_1,R_2$ of the same Abelian topological order, with a fixed common identification of their anyon labels and the same geometric seam $s$. Suppose each has an admissible lift as in Definition~\ref{def:register}, inducing $P_i=\rho_{R_i}(s)$. The matching rule for a string across the seam in realization $R_i$ is $b\mapsto P_i b$. If $P_1\ne P_2$, the two seams have different matching rules in this fixed labelling. An orientation-reversing loop with no other label monodromy can close with a fixed label only for $b\in\ker(1-P_i)$. These statements hold at every size at which the stated local realization and lift exist; they require no enumeration of finite lattices.
\end{proposition}
\begin{proof}
Conjugating a string segment by $\widetilde s_{R_i}$ transports its endpoint charge to $P_i b$ by the definition of the induced action on string sectors. Gluing matches the endpoint to this transported sector. Fusion of strings makes $P_i$ a homomorphism, and reversing the crossing gives its inverse. If $P_1\ne P_2$, some label has different images, so the marked matching rules differ. With no other monodromy, closing the transported label requires $P_i b=b$, which is precisely the stated kernel condition. This proves the transport statement, not the existence or dimension of a global frustration-free code space.
\end{proof}

\begin{remark}[Proof, grid checks, and open scope]\label{rem:register-scope}
\emph{Proved.} Proposition~\ref{prop:register} is a size-independent statement about label transport under its hypotheses. It does not classify all microscopic seams. In particular, different matrices in different label conventions need not describe inequivalent seams: a relabelling $Q$ changes $P$ to $QPQ^{-1}$. Nor does $P$ specify the phases of stabilizer relations.

\emph{Verified on the grid.} The concrete degeneracy entries of Table~\ref{tab:gsd} retain their stated ranges, and the distance and gate data retain the ``at the sizes computed'' qualifications of Sec.~\ref{sec:codes}. 

\emph{Conditional comparison.} The parity-invariant-anyon count \cite{ChanTeoRyu2016}, written as
\[
 \GSD(K;P)=|\ker(1-P)|,
\]
is a separate comparison for a consistent parity-twisted topological code to which that counting construction applies. It is not a consequence of the transport proof alone, nor a formula for the ground-state multiplicity of an arbitrary glued Hamiltonian. A microscopic application must establish a nonempty code space and the applicability of the topological count; the commuting but sign-frustrated odd-$L_x$ site gluing below fails the first condition. The equality for the local $u$-family, with its relation-phase conditions, is discussed separately in an unpublished companion manuscript~\cite{Wang2026paperE}; it is not an input to the results established here.

\emph{Open here.} We do not prove that arbitrary register realizations are related by a local unitary that also intertwines their symmetry actions, that $P$ is a complete invariant of a seam, or that it alone determines commutativity, frustration, distance or logical gates. Additional microscopic and symmetry data may distinguish realizations with the same label map. The literal noncommuting seam of Sec.~\ref{sec:seam} lies outside the admissible-lift hypothesis; its obstruction is established there by its Pauli algebra.
\end{remark}

\paragraph{Notation for the $\Z_N$ extension.}
For comparison, the $\Z_N$ extension is discussed in Junhao (Eigen) Wang's unpublished companion manuscript, \emph{Cross-caps do not supply magic within stabilizer operations: local seams and $\mathbb Z_N$ codes on lattice Klein bottles}~\cite{Wang2026paperE}. The present $\Z_2$ derivations and finite checks do not depend on that manuscript's conclusions. Here $N$ denotes the qudit dimension; in $\Had^{\otimes N}$ it counts physical qubits. In its fixed convention, $\mathcal A=\Z_N^2$, $u\in\Z_N^\times$ and $C=-I$. For upward seam transport of the ordered label column $(A,B)^T$, the single-site Weyl-monomial far-corner $u$-family has, for even $L_x$ and $L_x,L_y\ge3$, a diagonal branch ($a+L_y$ odd) and an exchange branch ($a+L_y$ even), with respective maps
\[
\begin{gathered}
 P_d=\begin{pmatrix}-u&0\\0&u^{-1}\end{pmatrix},\qquad
 P_e=\begin{pmatrix}0&-u\\u^{-1}&0\end{pmatrix},\\
 P_e^2=C.
\end{gathered}
\]
These are the conventions used for the local-gluing and seam-action results in that unpublished manuscript; the microscopic $u$-lift preserves the check constraints but need not preserve the chosen Hamiltonian for general $u$. At $N=2$ the allowed unit is $u=1$, so $P_d=I$, $P_e$ exchanges $e$ and $m$, and $C=I$. The edge-register map is $P=\mathrm{diag}(1,-1)$ in the same convention and is trivial at $N=2$. The two-colour matrices above assume even $L_x$; they are not assigned to the odd-$L_x$ site quotient. This is a comparison of already stated seam actions, not an extension of any numerical grid in this paper.

\begin{table}[!t]
\caption{Ground-state degeneracy of the consistent gluings (GF(2) rank, with every relation checked to have phase $+1$), the sign-frustrated dressed gluing at odd $L_x$, and the obstruction of the literal reflection gluing in the site register. ``all'' refers to the sizes listed in Appendix~\ref{app:methods}; $a$ is the axis offset, $h=L_x/2$ the glide shift. $^{*}$The plaquettes commute, but their product is $-1$: no state satisfies all of them, and $H=-\sum_pF_p$ has $E_0=-N_p+2$ with a $2N_p$-fold ground space (Sec.~\ref{sec:register}).}
\label{tab:gsd}
\centering\scriptsize\setlength{\tabcolsep}{1.5pt}
\begin{tabular}{@{}lllp{0.26\columnwidth}@{}}
\toprule
register & gluing & consistent & $\GSD$\\
\midrule
edge & torus & yes & $4$\\
edge & sheared torus & yes & $4$\\
edge & mirror, any $a$, any $L_y$ & yes & $4$\\
site & torus, $L_x,L_y$ even & yes & $4$\\
site & torus, $L_x$ or $L_y$ odd & yes & $2$\\
site & sheared torus & yes & $2$ iff $L_x/2{+}L_y$ odd\\
site & mirror dressed, $L_x$ even & yes & $2$ iff $a{+}L_y$ even\\
site & mirror dressed, $L_x$ odd & no$^{*}$ & --- (frustrated)\\
site & glide, dressed & yes & $2$ iff $h$ even\\
site & mirror, literal & no: $2L_x$ pairs & --- (Sec.~\ref{sec:seam})\\
site & glide, literal & no: $2L_x$ pairs & --- (Sec.~\ref{sec:seam})\\
\bottomrule
\end{tabular}
\end{table}

Table~\ref{tab:gsd} collects the $\mathrm{GF}(2)$ results. In the edge register every gluing we tried---the torus, the sheared torus, and the mirror Klein bottle with offsets $a=0,1$ at $L_y=4,5,6$ and $L_x=4,6,8$---closes into a stabilizer code of rank $n-2$, $\GSD=4$. For the Klein bottle this is $|H_1(K;\Z_2)|=4$: the reflection does not act on the anyon labels, and the code is the ordinary $\Z_2$ homology code of the surface. In the site register the torus has $\GSD=4$ for even$\times$even sizes and $2$ when a side is odd, as is known \cite{Wen2003plaquette}: the colouring of the plaquettes is inconsistent around an odd cycle, which exchanges $e$ and $m$. The sheared torus has $\GSD=2$ iff $L_x/2+L_y$ is odd, for the same reason: the plaquette of row $L_y-1$ at column $x$ and the plaquette of row $0$ that it touches across the wrap, at column $s(x)=x+L_x/2$, have colours $x+L_y-1$ and $x+L_x/2$, which coincide---i.e.\ the checkerboard is broken across the seam---exactly when $L_x/2+L_y$ is odd.

\paragraph{The dressed reflection is the symmetry.} Before gluing, we ask which lattice operations are symmetries of the site-register model on the $4\times4$ parent torus, by applying each candidate to a random state vector and comparing $UH$ with $HU$. The translations are symmetries (deviation $3\times10^{-16}$). The bare reflections $x\mapsto -x$, $x\mapsto1-x$, $y\mapsto-y$ and the bare glide $(x,y)\mapsto(x+2,-y)$ are not (relative deviation $1.38$); each of them composed with $\Had^{\otimes N}$ is a symmetry (deviation $\le6\times10^{-16}$), while $\Had^{\otimes N}$ composed with a translation is not ($1.44$). The reason is visible in Eq.~\eqref{eq:wen}: a reflection maps the corner pattern $X,Z,X,Z$ of a plaquette onto the pattern $Z,X,Z,X$ of its image, which the global Hadamard restores. On the four-dimensional ground space the translations act with eigenvalues $\{1,1,1,-1\}$ (trace $2$), the representation of $e\leftrightarrow m$ on the logical states; $\Had^{\otimes N}\circ(x\mapsto-x)$, $\Had^{\otimes N}\circ(y\mapsto-y)$ and the dressed glide act with the same eigenvalues, while $\Had^{\otimes N}\circ(x\mapsto1-x)$ acts as the identity (trace $4$). The dressed reflection therefore permutes anyons or not according to the parity of its axis offset, even on the torus: the reflection with $a$ odd maps the plaquette $(x,y)$ to the plaquette with corner $(a-x-1,y)$, of colour $a-1-x+y\equiv x+y\pmod2$, and preserves colours, whereas $a$ even exchanges them.

\paragraph{The colour-holonomy rule.} The dressed mirror plaquettes commute for every size and offset we tested ($L_x=4,\dots,12$; $L_y=3,\dots,6$; $a=0,1,2,3$; twenty-one sizes, Appendix~\ref{app:methods}), and the product of all of them is $(-1)^{L_x}$ times the identity. For even $L_x$ every relation has phase $+1$, the gluing closes into a code, and
\begin{equation}
\GSD=\begin{cases}2,& a+L_y\ \text{even},\\ 4,& a+L_y\ \text{odd},\end{cases}
\label{eq:rule}
\end{equation}
and the dressed glide quotient closes with $\GSD=2$ for $h$ even and $4$ for $h$ odd ($L_x=4,\dots,12$). For odd $L_x$ the single relation among the dressed plaquettes has phase $-1$ ($L_x=5,7,9,11$, all offsets): no state satisfies all of them, and $H=-\sum_pF_p$ is frustrated by a global sign, with one plaquette necessarily violated, $E_0=-N_p+2$, and a $2N_p$-fold degenerate ground space (the violated plaquette can be any of the $N_p$, and the remaining $N_p-1$ independent plaquettes leave a two-dimensional joint eigenspace on the $N_p$ qubits), as dense diagonalization confirms at $3\times3$, $3\times4$ and $5\times2$ and a Lanczos run at $5\times3$ ($E_0=-13$); reversing the sign of one plaquette term restores a consistent code with $\GSD=2$. This is the Klein-bottle counterpart of Wen's odd$\times$odd torus at the opposite sign of the coupling \cite{Wen2003plaquette}, where the same kind of constraint forces one vortex into the ground state (with our sign the odd$\times$odd torus itself is a consistent code, Table~\ref{tab:gsd}; on the Klein bottle with odd $L_x$ the standard sign is the frustrated one). (An earlier draft reported $\GSD=2$ for these sizes from the $\mathrm{GF}(2)$ rank alone; the sign was found by an independent $\Z_N$ implementation of the same models.) Both even-$L_x$ rules are the sheared-torus rule in disguise: the plaquette of row $L_y-1$ at column $x$ has colour $x+L_y-1$, the plaquette of row $0$ that it touches through the reflection has its SW corner at $a-x-1$ and colour $a-x-1$, and the two coincide---the checkerboard is broken, the non-orientable cycle acts as $e\leftrightarrow m$---exactly when $a+L_y$ is even; for the glide the plaquettes $(x,0)$ and $(x+h,0)$ are neighbours across the half row and have equal colours exactly when $h$ is even. The Hadamard dressing itself does not permute the anyons: it is what makes the reversed corner pattern consistent, and the permutation is carried by the lattice map. The values $4$ and $2$ are the numbers of anyons of the toric code left invariant by a trivial and by an $e\leftrightarrow m$ parity action, in agreement with the general count for topological order on non-orientable surfaces \cite{ChanTeoRyu2016}; the seam of the dressed gluing with $a+L_y$ even is a closed duality twist line.

\paragraph{Application and limits of the transport statement.} Proposition~\ref{prop:register} organizes the comparison on this lattice Klein bottle. In the edge register the seam acts trivially on the $\Z_2$ anyon labels, with $\GSD=4$ as stated above. In the site register the reflection requires the Hadamard dressing to preserve the plaquette checks, but the dressing alone does not force $e\leftrightarrow m$: for even $L_x$ the colour holonomy $a+L_y$ determines whether $P=P_e$ or $P=P_d$, with the reported code degeneracy $2$ or $4$ in Eq.~\eqref{eq:rule}. Odd $L_x$ remains the sign-frustrated exception: the dressed plaquettes commute but admit no common $+1$ eigenstate. Thus the label-transport proposition and the phase-sensitive microscopic consistency test answer different questions. The literal Pauli gluing in the site register, which tries to identify the reversed pattern without the dressing, is inconsistent for every size and offset ($L_x=4,\dots,12$, even and odd; $L_y=3,\dots,6$; $a=0,\dots,3$): exactly $2L_x$ pairs of plaquette operators anticommute, and we now turn to the Hamiltonian this defines.

\section{The obstructed gluing: locked Majorana rings}\label{sec:seam}

\begin{figure*}[!t]
\centering\includegraphics[width=0.98\textwidth]{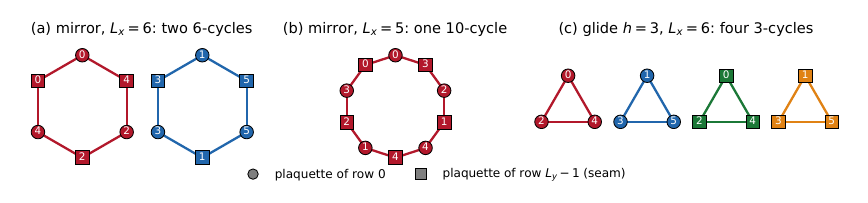}
\caption{Anticommutation graphs of the literal gluing (vertices: plaquette operators, edges: anticommuting pairs; numbers: column). (a) Mirror Klein bottle, $L_x=6$: the plaquettes of row $L_y-1$ (squares) and of row $0$ (circles) form two $6$-cycles, one per column parity. (b) Mirror Klein bottle, $L_x=5$: the two cycles merge into one $10$-cycle through all columns. (c) Glide quotient with $h=3$: four $3$-cycles, two in each row.}
\label{fig:graph}
\end{figure*}

\paragraph{Anatomy of the obstruction.} The Hamiltonian $H=-\sum_pF_p$ of the literal gluing is a sum of $N_p=L_xL_y$ four-body Pauli terms that no longer commute. The anticommutation graph depends only on the parity of $L_x$ (Fig.~\ref{fig:graph}): the $L_x$ plaquettes of row $L_y-1$ (the seam plaquettes) and the $L_x$ plaquettes of row $0$ each anticommute with exactly two operators, and the $2L_x$ involved operators form cycles that alternate between the two rows. The seam plaquette at column $x$ shares one site with the row-$0$ plaquettes whose corners it meets with the wrong Pauli, at columns $a-x-2$ and $a-x$, and two sites with the plaquette at $a-x-1$, with which it commutes; this is why each operator has degree two, and following the cycle from the seam plaquette at $x$ through the row-$0$ plaquette at $a-x$ leads to the seam plaquette at $x-2$. For even $L_x$ the step of two columns preserves the column parity and there are two disjoint cycles of length $L_x$, one through the even and one through the odd columns (Fig.~\ref{fig:graph}a); for odd $L_x$ the step generates all of $\Z_{L_x}$ and the two cycles merge into a single cycle of length $2L_x$ (Fig.~\ref{fig:graph}b). The $\mathrm{GF}(2)$ rank of the adjacency matrix is $2L_x-4$ for even $L_x$ (each even cycle has corank $2$) and $2L_x-2$ for odd $L_x$ (one even cycle), and all other $N_p-2L_x$ plaquettes commute with everything: we call them the bulk stabilizers. For the glide quotient with $h$ even the graph is again two $L_x$-cycles; for $h$ odd it is four cycles of length $h$, two within each of the rows $0$ and $L_y-1$ (Fig.~\ref{fig:graph}c), with the same rank $2L_x-4$ (each odd cycle has corank $1$).

\begin{lemma}[cycle algebra]\label{lem:cycle}
Let $O_1,\dots,O_L$ be Hermitian involutions such that $O_j$ anticommutes with $O_{j\pm1}$ (indices mod $L$) and commutes with all others. For $L$ even, the products $E=\prod_{j\,\mathrm{even}}O_j$ and $D=\prod_{j\,\mathrm{odd}}O_j$ are central, $ED=\pm\Pi$ with $\Pi=O_1O_2\cdots O_L$, and in every irreducible representation $O_j=i\gamma_j\gamma_{j+1}$ for Majorana operators $\gamma_1,\dots,\gamma_L$ with $\gamma_{L+1}=\pm\gamma_1$; since $\Pi=i^L\gamma_1\gamma_{L+1}$, the two eigenvalues $\pm i^L$ of $\Pi$ label the periodic and the antiperiodic ring. Consequently $H_{\rm ring}=-\sum_jO_j$ is a free Majorana ring with periodic or antiperiodic boundary condition, with single-particle energies $2\sin k$ at $k=2\pi m/L$ or $(2m+1)\pi/L$, and
\begin{equation}
E_0^{\rm AP}(L)=-\frac{2}{\sin(\pi/L)},\qquad E_0^{\rm P}(L)=-\frac{2}{\tan(\pi/L)},
\label{eq:ring}
\end{equation}
the antiperiodic ground state being unique and the periodic one two-fold degenerate (two zero modes). For $L$ odd the product $\Pi$ is central and anti-Hermitian, both of its sectors have the same spectrum, and $E_0=-\cot(\pi/2L)$ ($=-\sqrt3$ for $L=3$).
\end{lemma}
The representation statement is the standard one for a Clifford-type algebra with the cycle as its anticommutation graph; the spectra follow from the single-particle problem and are confirmed by direct diagonalization of the algebra in a Jordan--Wigner representation for $L=4,\dots,10,14$ and $L=3,5,7$ (Appendix~\ref{app:cycles}). For the seam this means: for even $L_x$ each of the two cycles is a Majorana ring of $L_x$ sites, for odd $L_x$ the seam is a single Majorana ring of $2L_x$ sites, and the ground energy of a ring is $E_0^{\rm AP}$ or $E_0^{\rm P}$ according to the value of its product $\Pi_r$.

\paragraph{Sector locking.} Which sector occurs is not a choice. The ring products $\Pi_r$ are central and commute with $H$; their product is, up to a sign fixed by the anticommutation relations, the product of all seam-adjacent plaquettes, which equals the product of all plaquettes divided by the product of the bulk stabilizers. The product of all $N_p$ plaquettes is the identity up to a sign (every site carries $X$ twice and $Z$ twice), so $\prod_r\Pi_r$ is a fixed sign times the product of the bulk stabilizers, which is $+1$ in every state that satisfies them. For even $L_x$ this ties the sectors of the two rings together; for odd $L_x$ the product of the single ring is itself fixed, and the ring cannot change sector without a violated bulk plaquette. The exact Pauli trace of the sector projectors (see \emph{Degeneracy} below) makes this quantitative at every size we computed: in the space where all bulk stabilizers are $+1$, the $(\mathrm{AP},\mathrm{AP})$ and $(\mathrm{P},\mathrm{P})$ sectors of an even-$L_x$ seam have equal dimension and the mixed sectors have dimension zero, while for odd $L_x$ the antiperiodic sector has the full dimension and the periodic sector has dimension zero. On the $4\times3$ lattice we also evaluate the central elements on the exact eigenstates: $\Pi_1=\Pi_2=-1$ (the antiperiodic value for $L_x=4$) on all ground states, $+1$ on all $32$ states of the first excited multiplet, and the four alternating products $E_r,D_r$ are $+1$ on the ground space (they take both signs on the excited multiplet). For even $L_x$ the ground state is the $(\mathrm{AP},\mathrm{AP})$ sector, the first excited multiplet is $(\mathrm{P},\mathrm{P})$, and the mixed sectors require a violated bulk stabilizer; for odd $L_x$ the ground state is the antiperiodic sector of the $2L_x$-ring and the periodic sector is reached only together with a violated bulk stabilizer.

\paragraph{Exact ground energy, gap and levels.} With $N_p-2L_x$ bulk stabilizers at $-1$ and the seam rings antiperiodic,
\begin{equation}
E_0=-N_p+2L_x-\begin{cases}4/\sin(\pi/L_x), & L_x\ \text{even},\\ 2/\sin(\pi/2L_x), & L_x\ \text{odd},\end{cases}
\label{eq:E0}
\end{equation}
i.e.\ $2E_0^{\rm AP}(L_x)$ or $E_0^{\rm AP}(2L_x)$ added to the bulk energy. For even $L_x\ge4$ the lowest excitations are, in order of appearance: the $(\mathrm{P},\mathrm{P})$ sector at
\begin{equation}
\Delta=2\big[E_0^{\rm P}(L_x)-E_0^{\rm AP}(L_x)\big]=4\tan\frac{\pi}{2L_x} ,
\label{eq:gap}
\end{equation}
one single-particle excitation of an antiperiodic ring at $4\sin(\pi/L_x)$---present only when $a+L_y$ is odd, i.e.\ when the plaquettes obey a single $\mathrm{GF}(2)$ relation: for $a+L_y$ even the second relation fixes the alternating products $E_r,D_r$, which are the fermion parities of the two rings, and a single mode cannot be excited without a bulk flip---and a mixed sector made possible by one violated bulk plaquette at $2+2\tan(\pi/2L_x)$. For odd $L_x$ the periodic sector is not available without a bulk flip, and the lowest excitation is a single mode of the $2L_x$-ring,
\begin{equation}
\Delta=4\sin\frac{\pi}{2L_x} ,
\label{eq:gapodd}
\end{equation}
followed by the periodic sector with one violated bulk plaquette at $2+2\tan(\pi/4L_x)$ and by two ring modes at $8\sin(\pi/2L_x)$. Table~\ref{tab:seam} and Fig.~\ref{fig:spectra}(b) compare these with exact diagonalization on $4\times3$, $4\times4$, $6\times3$ and $5\times3$ for both offsets and on $5\times4$ and $7\times3$ ($21$ qubits). The ground energies agree with Eq.~\eqref{eq:E0} to $2\times10^{-13}$ ($E_0=-9.656854$, $-13.656854$, $-14$; $-7-2\sqrt5=-11.472136$ at $5\times3$, $-16.472136$ at $5\times4$, $-15.987918$ at $7\times3$), the gap is $4(\sqrt2-1)=1.656854$ for $L_x=4$, $4(2-\sqrt3)=1.071797$ for $L_x=6$, $\sqrt5-1=1.236068$ for $L_x=5$ and $0.890084$ for $L_x=7$, and the next levels are $2\sqrt2$ ($L_x=4$; here $4\sin(\pi/4)$ and $2+2\tan(\pi/8)$ coincide), $2$ and $2.5359$ for $6\times3$ at $a=0$ but $2.5359$ alone at $a=1$, where $a+L_y$ is even and the single-mode level is absent, $2.316769$ and $2\sqrt5-2=2.472136$ for $L_x=5$ (the periodic sector with a bulk flip, then two ring modes), and $1.780167$ and $2.225346$ for $L_x=7$ (the same two levels in the opposite order), all as predicted. The two $24$-qubit predictions of Fig.~\ref{fig:spectra}(c), $E_0=-24+12-8=-20$ for the $6\times4$ mirror gluing and $E_0=-24+12-4\sqrt3=-18.928203$ for the $6\times4$ glide quotient with $h=3$ (four $3$-cycles at $-\sqrt3$ each), are reproduced by a Lanczos iteration to $5\times10^{-14}$.

\begin{table}[!t]
\caption{The obstructed gluing: exact diagonalization against the closed forms (in every entry with an exact-diagonalization energy, ED and formula agree to $2\times10^{-13}$). $E_0$ from Eq.~\eqref{eq:E0} (mirror, and glide with $h$ even) or $-N_p+2L_x-4\cot(\pi/2h)$ (glide, $h$ odd); $\Delta$ from Eq.~\eqref{eq:gap}, or from Eq.~\eqref{eq:gapodd} for odd $L_x$ ($^{*}$); $\GSD$ from the exact Pauli trace of the ground projector (Sec.~\ref{sec:seam}), with the count found by exact diagonalization in parentheses when the Krylov space was large enough to hold the whole multiplet.}
\label{tab:seam}
\centering\scriptsize\setlength{\tabcolsep}{1.5pt}
\begin{tabular}{llrrr}
\toprule
gluing & size & $E_0$ & $\Delta$ & $\GSD$\\
\midrule
mirror $a{=}0$ & $4\times3$ & $-9.656854$ & $1.656854$ & $8$ (8)\\
mirror $a{=}1$ & $4\times3$ & $-9.656854$ & $1.656854$ & $16$ (16)\\
mirror $a{=}0$ & $4\times4$ & $-13.656854$ & $1.656854$ & $16$ (16)\\
mirror $a{=}1$ & $4\times4$ & $-13.656854$ & $1.656854$ & $8$ (8)\\
mirror $a{=}0$ & $6\times3$ & $-14$ & $1.071797$ & $32$ (32)\\
mirror $a{=}1$ & $6\times3$ & $-14$ & $1.071797$ & $64$\\
mirror $a{=}0$ & $6\times4$ & $-20$ (Lanczos) & --- & $64$\\
mirror $a{=}1$ & $6\times4$ & --- & --- & $32$\\
mirror $a{=}0/1$ & $8\times3$ & --- & --- & $128/256$\\
mirror $a{=}0/1$ & $8\times4$ & --- & --- & $256/128$\\
mirror $a{=}0$ & $5\times3$ & $-11.472136$ & $1.236068^{*}$ & $32$ (32)\\
mirror $a{=}1$ & $5\times3$ & $-11.472136$ & $1.236068^{*}$ & $32$ (32)\\
mirror $a{=}0$ & $5\times4$ & $-16.472136$ & $1.236068^{*}$ & $32$\\
mirror $a{=}0$ & $7\times3$ & $-15.987918$ & $0.890084^{*}$ & $128$\\
mirror $a{=}0/1$ & $7\times4$ & --- & --- & $128/128$\\
mirror $a{=}0/1$ & $9\times3$ & --- & --- & $512/512$\\
glide $h{=}2$ & $4\times3$ & $-9.656854$ & $1.656854$ & $16$ (16)\\
glide $h{=}3$ & $6\times3$ & $-12.928203$ & $2$ & $256$\\
glide $h{=}3$ & $6\times4$ & $-18.928203$ (Lanczos) & --- & $256$\\
glide $h{=}5$ & $10\times4$ & --- & --- & $4096$\\
\bottomrule
\end{tabular}
\end{table}

\begin{figure*}[!t]
\centering\includegraphics[width=0.98\textwidth]{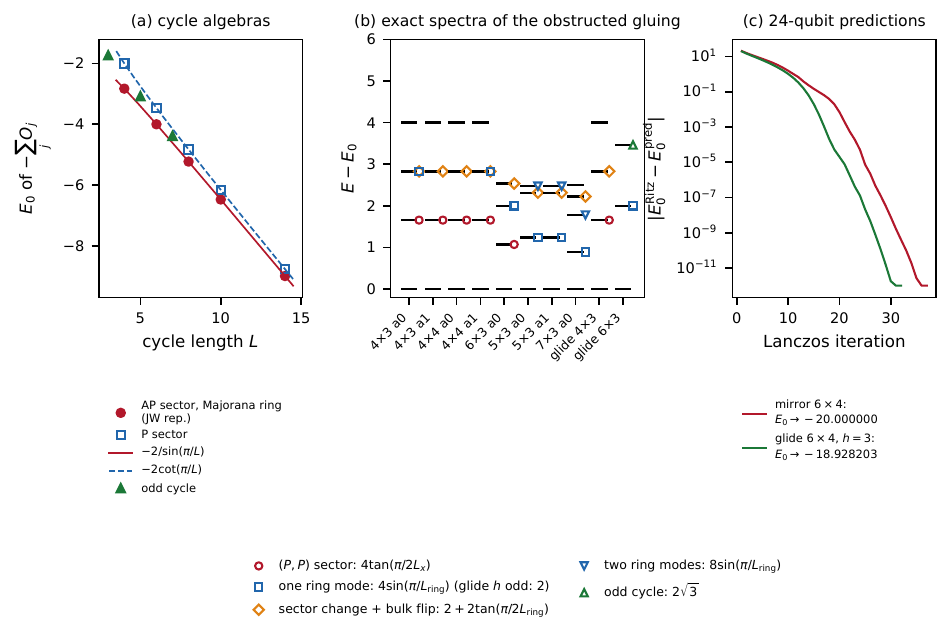}
\caption{(a) Ground energy of the cycle Hamiltonian $-\sum_jO_j$ in the two sectors of the central product, from the Jordan--Wigner representation (symbols) against Eq.~\eqref{eq:ring} (lines); triangles: odd cycles, $-\cot(\pi/2L)$. (b) Low levels of the frustrated Hamiltonian relative to $E_0$ (black) for the mirror and glide gluings, with the decoded levels in terms of the ring length $L_{\rm ring}=L_x$ (even $L_x$) or $2L_x$ (odd $L_x$): $(\mathrm P,\mathrm P)$ sector $4\tan(\pi/2L_x)$ (even $L_x$ only), one ring mode $4\sin(\pi/L_{\rm ring})$ (shown only where the ring parities are free, i.e.\ one $\mathrm{GF}(2)$ relation), a sector change with one bulk flip $2+2\tan(\pi/2L_{\rm ring})$, two ring modes $8\sin(\pi/L_{\rm ring})$; for the odd-cycle glide, one bulk flip ($2$) and one cycle excitation ($2\sqrt3$). (c) Convergence of the $24$-qubit Lanczos energies to the two predictions.}
\label{fig:spectra}
\end{figure*}

\paragraph{Degeneracy.} The ground-state degeneracy follows from the same structure. The ground space is the common eigenspace in which every bulk stabilizer is $+1$, every central element of the seam algebra takes its ground value, and each ring is in the unique ground state of its irreducible representation. The seam algebra has $\mathrm{GF}(2)$ symplectic rank $\mathrm{rank}\,\mathsf A$, so its irreducible representations have dimension $2^{\frac12\mathrm{rank}\,\mathsf A}$; counting dimensions gives
\begin{equation}
\begin{aligned}
\GSD&=2^{\,n-\mathrm{rank}\,S+\frac12\mathrm{rank}\,\mathsf A}\\
&=2^{\,n-\mathrm{rank}\,S+L_x-2+[L_x\ \mathrm{odd}]},
\end{aligned}
\label{eq:gsd}
\end{equation}
where $S$ is the full set of $N_p$ plaquette rows and $n-\mathrm{rank}\,S$ is the number of $\mathrm{GF}(2)$ relations among them: for even $L_x$ it is $2$ exactly when $a+L_y$ is even (the same parity as in Eq.~\eqref{eq:rule}) and $1$ otherwise; for odd $L_x$ it is $1$ for every offset. Explicitly, $\GSD=2^{L_x-1}$ or $2^{L_x}$ for even $L_x$ and $2^{L_x}$ for odd $L_x$. For the glide quotient with $h$ odd the two sectors of each odd cycle are degenerate (Lemma~\ref{lem:cycle}), and only the product of the four cycle products is fixed by the bulk, so Eq.~\eqref{eq:gsd} acquires a factor $2^{3}$: $\GSD=2^{L_x+2}$. We verify these counts exactly, without diagonalization, by writing the ground projector as $\prod_{\rm bulk}\tfrac12(1+F_p)\prod_{\rm cycles}P_c$ and taking its trace in the Pauli basis ($\Tr$ of a Pauli string is $2^n$ for the identity and $0$ otherwise). The ground projector $P_c$ of a cycle in a given sector is expanded exactly in the $2^{L_c-1}$ ordered products $\prod_{j\in T}O_j$, $T\subseteq\{1,\dots,L_c-1\}$, which are orthogonal Pauli strings spanning the cycle algebra of that sector, with coefficients read off in the Jordan--Wigner representation, multiplied by the central idempotent $\tfrac12(1+\Pi_c/\pi_c)$ that selects the sector; the same trace gives the dimension of every sector combination in the bulk-stabilized space (the locking statements above) and $\Tr(HP)/\Tr P$, the exact ground energy. The result is $8,16,16,8,32,64,64,32,128,256,256,128$ for the mirror gluing at $4\times3$, $4\times4$, $6\times3$, $6\times4$, $8\times3$, $8\times4$ (offsets $0/1$), $32,32,128,128,512$ at $5\times3$, $5\times4$, $7\times3$, $7\times4$, $9\times3$ (both offsets), $16$ for the glide at $4\times3$ and $4\times4$, $256$ for $6\times3$ and $6\times4$ with $h=3$ and for $8\times4$ with $h=4$, and $4096$ for $10\times4$ with $h=5$, all integers to $10^{-11}$ and all equal to the formula. Exact diagonalization finds the full multiplet where the Krylov space allows it ($8,16,16,8$ at $L_x=4$; $32$ at $6\times3$, $a=0$, with $100$ Lanczos vectors; $32$ at $5\times3$ for both offsets) and a lower bound otherwise.

\paragraph{What the degeneracy is not.} The factor $2^{\frac12\mathrm{rank}\,\mathsf A}$ is a property of the frustrated Hamiltonian, not a topological degeneracy. On $4\times3$ we add single-site fields $h\sum X$, $h\sum Y$ or $h\sum Z$, uniform or random, to the two seam rows or to the bulk row: fields on the bulk row leave the eight-fold ground multiplet degenerate to $2\times10^{-13}$, while fields on the seam rows split it at second order---for a random $Z$ field the splitting is $9.3\times10^{-5},3.7\times10^{-4},1.5\times10^{-3},5.9\times10^{-3}$ at $h=0.01,0.02,0.04,0.08$ (log--log slope $2.0$), with the same slope for random $X$ and $Y$ fields---and no first-order term. The seam rings are exact and their spectrum is exact, but their degeneracy is lifted at order $h^2$ by every single-site seam field we tested; we make no protection claim, and we have not examined multi-site perturbations.

\section{Klein-bottle codes: logical operators and lattice-symmetry gates}\label{sec:codes}

Each consistent gluing of Secs.~\ref{sec:model}--\ref{sec:register} is a stabilizer code, and the lattice Klein bottle has a symmetry that the torus does not have: its half-period translation along the seam is not isotopic to the identity. In this section we compute exactly, for even $L_x$ (so that the half translation $h=L_x/2$ exists; the sizes are listed in Table~\ref{tab:codes}), the logical operators and distances of the three codes and the logical gate implemented by each enumerated deck-compatible affine isometry, with or without the Hadamard dressing. The distance formulas below are those observed at the sizes computed, not theorems. The result is register dependent once more. On the untwisted Klein bottle of the site register the isometries, acting as qubit relabellings and single-row Hadamards, implement the logical $\mathrm{SWAP}$ and the logical $\Had\otimes\Had$ of the two logical qubits---the group generated by the mapping-class action on $H_1(K;\Z_2)$ and by $e\leftrightarrow m$; on the edge register they implement only the $\mathrm{SWAP}$; and on the twisted code each audited isometry acts as a logical Pauli.

\paragraph{Homology and the cross-cap basis.} $H_1(K;\Z_2)=\Z_2^2$ is generated by the two-sided cycle $\gamma_x$ (a row) and the one-sided cycle $\gamma_y$ (a column through the axis, closed by the reflection), with intersection form $\gamma_x\cdot\gamma_x=0$, $\gamma_x\cdot\gamma_y=1$, $\gamma_y\cdot\gamma_y=1$: a one-sided curve meets its own push-off once. The classes $c_1=\gamma_y$ and $c_2=\gamma_x+\gamma_y$, represented by the columns through $x=a/2$ and $x=a/2+h$ (for $a$ even; for $a$ odd the columns next to the axis, with one horizontal step), are both one-sided and satisfy $c_i\cdot c_j=\delta_{ij}$: this is $K=\mathbb{RP}^2\#\mathbb{RP}^2$ with one cross-cap per $c_i$, and the half translation $x\mapsto x+h$ exchanges the two. The natural logical basis of the $\Z_2$ toric code on $K$ is the \emph{cross-cap basis}: $\bar Z_i$ the $e$ loop and $\bar X_i$ the $m$ loop along $c_i$. The two operators of one cross-cap anticommute because they cross once, those of different cross-caps commute, and the loops along $\gamma_x$ are products, $e(\gamma_x)=\bar Z_1\bar Z_2$. In this basis the two topological operations available on $K$ take their simplest form: the Dehn twist about $\gamma_x$, which acts on homology as $\gamma_y\mapsto\gamma_x+\gamma_y$ and generates the mapping-class action on $H_1(K;\Z_2)$ \cite{Lickorish1963}, exchanges $c_1\leftrightarrow c_2$ and is the logical $\mathrm{SWAP}$; the $e\leftrightarrow m$ exchange maps $\bar Z_i\leftrightarrow\bar X_i$ and is the logical $\Had\otimes\Had$, the observation behind the cross-cap Hadamard of Ref.~\cite{Kobayashi2024crosscap}. On the torus, by contrast, the $e$ and $m$ loops along the same cycle commute: pairing $e(\gamma_x)$ with $m(\gamma_y)$ as one logical qubit and $m(\gamma_x)$ with $e(\gamma_y)$ as the other, $e\leftrightarrow m$ is the $\mathrm{SWAP}$ of the two qubits, and with $\bar Z_i$ the $e$ loops on both qubits, the convention we keep below, it is $\mathrm{SWAP}\circ(\Had\otimes\Had)$; either way it is the operation with trace $2$ found in Sec.~\ref{sec:register}.

\paragraph{Codes and distances.} The logical operators are computed as a symplectic basis of $N(S)/S$ over $\mathrm{GF}(2)$ and the distance as the minimum weight over all non-trivial logical classes, by enumerating the whole stabilizer group (site register, $n\le24$), the $X$- and $Z$-type parts separately (edge register, a CSS code, $n\le48$), or all Pauli strings of weight $\le5$ (site register at $6\times5$). Table~\ref{tab:codes} collects the results and Fig.~\ref{fig:codes} shows minimum-weight representatives. The edge register gives $[[2L_xL_y,2,\min(L_x,L_y)]]$ with a type asymmetry, $d_Z=\min(L_x,L_y+[a\ \text{odd}])$ and $d_X=\min(L_x,L_y+[a\ \text{even}])$: the dual lattice of a Klein bottle with a site-centred axis has a bond-centred one, so on each cross-cap exactly one of the two loop types needs the extra step. The untwisted site code is $[[L_xL_y,2,\min(L_x,L_y)]]$: its strings run along diagonals, and both loop types close through the seam in $L_y$ sites for either axis parity. The twisted site code is $[[L_xL_y,1,\min(L_x,2\lceil L_y/2\rceil)]]$ on all sizes computed: its logical $\bar X$ is the $e$ loop along $\gamma_x$ (equivalent, modulo stabilizers, to the $m$ loop, since the seam exchanges them) and $\bar Z$ is the $\psi=e\times m$ loop along $\gamma_y$, a column of $Y$'s closed through the axis (weight $L_y$ for $L_y$ even, $L_y+1$ for $L_y$ odd; the twist requires $a\equiv L_y\bmod2$). The distance is $4$ at the seven sizes with $\min(L_x,2\lceil L_y/2\rceil)=4$ and $6$ at $6\times5$, where an exhaustive search finds no logical operator of weight $\le5$ and the row loop $ZXZXZX$ has weight $6$.

\begin{table*}[!t]
\caption{Klein-bottle codes (even $L_x$). $d$: exact distance at each size $(L_x,L_y)$ computed (edge register: $d_X/d_Z$ for $a=0$, and $d_Z/d_X$ for $a=1$); the closed forms for $d$ summarize these sizes only. ``gates'': the logical action of the $8L_y$ enumerated affine-isometry candidates in the cross-cap basis ($\bar Z_i$ the $e$ loops), with $t$ the twist parity and $s$ the colour parity of the isometry defined in the text, modulo logical Paulis.}
\label{tab:codes}
\centering\scriptsize\setlength{\tabcolsep}{4pt}
\begin{tabular}{llp{6.1cm}l}
\toprule
code & $[[n,k,d]]$ & $d$ at the sizes computed & gates\\
\midrule
edge, mirror, $a=0,1$ & $[[2L_xL_y,\,2,\,\min(L_x,L_y)]]$ & $(4,3){:}\,3$, $(4,4){:}\,4$, $(6,3){:}\,3$, $(6,4){:}\,4$, $(8,3){:}\,3$; $\ d_X/d_Z=4/3$ for $L_y=3$, $4/4$ at $(4,4)$, $5/4$ at $(6,4)$ & $\mathrm{SWAP}^{t}$\\
site, dressed, $a+L_y$ odd & $[[L_xL_y,\,2,\,\min(L_x,L_y)]]$ & $(4,3){:}\,3$, $(4,4){:}\,4$, $(4,5){:}\,4$, $(4,6){:}\,4$, $(6,3){:}\,3$, $(6,4){:}\,4$, $(8,3){:}\,3$, $(6,5){:}\,5$ & $\mathrm{SWAP}^{t}(\Had\otimes\Had)^{s}$\\
site, dressed, $a+L_y$ even & $[[L_xL_y,\,1,\,\min(L_x,2\lceil L_y/2\rceil)]]$ & $4$ at the first seven sizes above, $6$ at $(6,5)$ & $1$ or $\bar P$\\
edge, torus & $[[2L_xL_y,\,2,\,\min(L_x,L_y)]]$ & $(4,4){:}\,4$, $(6,4){:}\,4$ & $1$\\
site, torus & $[[L_xL_y,\,2,\,\min(L_x,L_y)]]$ & $(4,4){:}\,4$, $(6,4){:}\,4$ & $[\mathrm{SWAP}\circ(\Had\otimes\Had)]^{s}$\\
\bottomrule
\end{tabular}
\end{table*}

\begin{figure*}[!t]
\centering\includegraphics[width=0.98\textwidth]{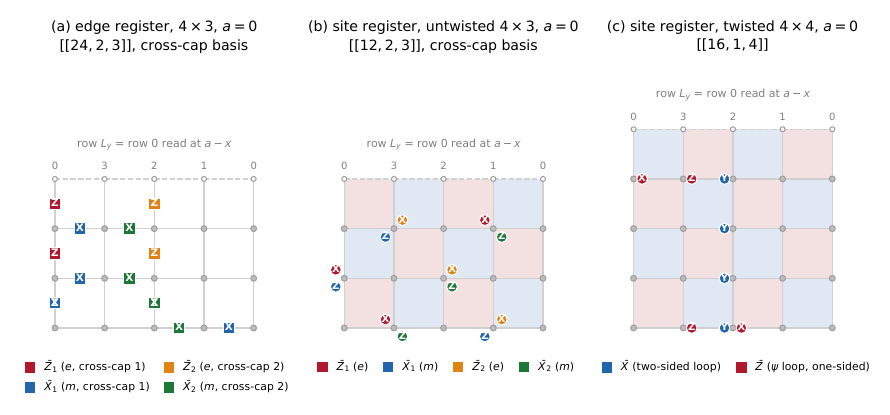}
\caption{Minimum-weight logical operators of the Klein-bottle codes (the glued row is drawn dashed, with the column index $a-x$ at which each site is read in row $0$). (a) Edge register, $4\times3$, $a=0$, cross-cap basis: $\bar Z_1,\bar Z_2$ are $Z$ strings on the columns $x=0$ and $x=2$ (the one-sided cycles $c_1,c_2$, weight $L_y=3$), $\bar X_1,\bar X_2$ are $X$ strings on the dual one-sided cycles, which need one horizontal step (weight $L_y+1=4$) because the dual axis is bond-centred. (b) Site register, untwisted $4\times3$, $a=0$: the $e$ loops ($\bar Z_i$, red/orange) are strings of $X$'s along one diagonal direction and the $m$ loops ($\bar X_i$, blue/green) strings of $Z$'s along the other; both close through the seam in $L_y$ sites. (c) Site register, twisted $4\times4$, $a=0$: $\bar X$ is a loop along the two-sided cycle (a zigzag of $Z$'s and $X$'s between the two seam rows) and $\bar Z$ a column of $Y$'s, the $\psi$ loop along the one-sided cycle.}
\label{fig:codes}
\end{figure*}

\paragraph{Audited affine isometries and code automorphisms.} The symmetry family considered here consists of maps $(x,y)\mapsto(s_xx+c,\,s_yy+c_y)$ of the covering square lattice, $s_x,s_y=\pm1$, that normalize the deck group generated by $(x,y)\mapsto(x+L_x,y)$ and $(x,y)\mapsto(a-x,y+L_y)$. Within this axis-preserving affine family, $s_x=+1$ requires $c\in\{0,h\}$, $s_x=-1$ requires $c\in\{a,a+h\}$, and $c_y$ is taken modulo $L_y$, because translation by $L_y$ acts on the quotient as reflection $x\mapsto a-x$. These are the $8L_y$ candidates used in the audit. The gate statements below refer to this enumerated family and the sizes of Table~\ref{tab:codes}. We do not assert that every automorphism of a finite quotient graph lifts to the covering square lattice, or classify additional graph symmetries at special sizes. In the edge register each listed isometry permutes the stabilizers and is an automorphism of the code as a qubit permutation. In the site register an isometry moves the seam, and the Hadamard dressing must move with it: the map composed with $\Had$ on the rows whose preimage crossed the seam (for the unit row shift, row $0$), and on all rows if the map reverses the orientation of the plane (a reflection reverses the corner pattern of every plaquette, Sec.~\ref{sec:register}), is an automorphism. We checked, with the full Pauli phases---each stabilizer generator must be mapped to a $+1$ element of the stabilizer group---that all $8L_y$ listed candidates are automorphisms of every code in Table~\ref{tab:codes}, with a unique union of rows for the dressing in the site register, which is the one just described in every case.

\paragraph{Logical action.} Two parities of an isometry decide what it does. The \emph{twist} $t=1$ if the map shifts the seam direction by half a period relative to the axis ($x\mapsto x+h$ or $x\mapsto a+h-x$), else $0$; the \emph{colour parity} $s=c+c_y-[s_x=-1]-[s_y=-1]\bmod2$ is the change of the plaquette colour $(x+y)\bmod2$ under the map (the image of the plaquette with south-west corner $(x,y)$ has its south-west corner at $(s_xx+c-[s_x=-1],\,s_yy+c_y-[s_y=-1])$). In the cross-cap basis, with $\bar Z_i$ the $e$ loops, every listed isometry acts on the untwisted site code as $\mathrm{SWAP}^{t}(\Had\otimes\Had)^{s}$ and on the edge code as $\mathrm{SWAP}^{t}$, up to logical Paulis, at every size and offset in Table~\ref{tab:codes}. The twist is the Dehn twist on homology: the half translation maps the column $c_1$ to the column $c_2$, i.e.\ acts on $H_1(K;\Z_2)$ as $\gamma_y\mapsto\gamma_x+\gamma_y$, and is therefore not isotopic to the identity---the translations of the torus all are, and on the torus the audited translation and reflection family acts trivially on the edge register (Table~\ref{tab:codes}). The colour parity is the anyon exchange: the maps that exchange the plaquette colours act as $\Had\otimes\Had$ on the Klein bottle and as $\mathrm{SWAP}\circ(\Had\otimes\Had)$ on the torus in the same $e$-loop convention, for the reason given above. In the edge register the colour parity plays no role: the stabilizers carry no checkerboard, none of the audited isometries exchanges $e$ and $m$, and the $e$--$m$ duality itself---the half-lattice shift followed by $\Had^{\otimes n}$---is not an automorphism: it maps the code with axis offset $a$ onto the code with offset $a-1$ (verified at all five sizes), an isomorphism between the two axis parities, because the dual of a Klein bottle with a site-centred axis has a bond-centred one and no affine map of $\Z^2$ conjugates one deck group into the other. On the twisted code every listed isometry acts as the identity or as one fixed logical Pauli $\bar P$: the unit row shift (with its row of Hadamards) and the row reflection $y\mapsto-y$ act as $\bar P$, the half translation as $\bar P$ for $h$ even and trivially for $h$ odd, and the group generated modulo Paulis is trivial. This is expected: both logical operators of the twisted code, the $\psi$ loop and the loop along the two-sided cycle, are invariant under $e\leftrightarrow m$, and the mapping-class action has nothing to exchange.

\paragraph{Confirmation.} For $n\le20$ we build the exact code space, evaluate the matrix of each of the $304$ automorphisms on it (unitary to $4\times10^{-13}$), and compare the conjugated logical operators with the Pauli-phase prediction (deviation $\le5\times10^{-13}$); on the code space the $\mathrm{SWAP}$ has eigenvalues $\{1,1,1,-1\}$ and $\Had\otimes\Had$ has $\{1,1,-1,-1\}$. In an abstract operation model that counts a qubit permutation or relabelling as one layer, these are depth-one Clifford operations with single-qubit Hadamards, and they preserve the weight of every error. This convention does not give a one-layer implementation by spatially local gates on a two-dimensional hardware layout; routing a permutation may require additional gates and depth. The group they generate modulo Paulis, $\Z_2\times\Z_2$, is the group generated by the mapping-class action on $H_1(K;\Z_2)$ and by the anyon symmetry, consistent with the general constraint that locality-preserving logical gates of a two-dimensional topological code act through the anyon symmetries and the mapping class group \cite{Beverland2016}: the audited untwisted site-code family realizes that group at the stated sizes. The cross-cap Hadamard of Ref.~\cite{Kobayashi2024crosscap} appears here as a row shift with one row of Hadamards, and the Dehn twist, which on the torus is a shear---a permutation of qubits over distances of order $L$ followed by a constant-depth correction \cite{Zhu2020dehn,Lavasani2019}---is on the Klein bottle a rigid lattice translation. No audited isometry gives a logical $S$, and the $\mathrm{SWAP}$ is a $\mathrm{CNOT}$ only in the homology basis $(\gamma_x,\gamma_y)$, a change of tensor-product structure rather than an entangling operation. The twisted code, with its closed duality line, is the poorer code of the two: one logical qubit, no non-Pauli gate from the audited isometries, and a distance limited by the $\psi$ loop.

\section{Chiral matter: a holonomy doublet, its logical operators and its splitting}\label{sec:matter}

The site register of Secs.~\ref{sec:register} and \ref{sec:codes} is the register of the $\pi$-flux $p$-wave superconductor of the companion paper~\cite{Wang2026paper1}. We now couple that matter to the gauge field and ask what becomes of the four Klein-bottle holonomy sectors. The answer, established exactly on three small lattices, is that the matter selects two of the four and leaves a two-fold degenerate doublet whose logical $X$ is the half translation of Sec.~\ref{sec:codes}, and whose splitting under an electric field is, on these lattices, of order $h^{\min(L_x,2L_y)}$: the two-sided cycle costs $L_x$ link flips, and the one-sided cycle, although only $L_y$ links long, must be traversed twice because a single traversal leaves the matter-selected pair. We state precisely what is computed and make no claim beyond the sizes reached.

\paragraph{Model.} On the $L_x\times L_y$ mirror Klein bottle with bond-centred axis, $(x,y+L_y)\sim(1-x,y)$, a $\Z_2$ gauge field $\sigma^z_l=\pm1$ lives on the links and spinless fermions on the sites,
\begin{align}
H={}&K\sum_pB_p-h\sum_l\sigma^x_l\nonumber\\
&+\sum_{\langle ij\rangle}\sigma^z_{ij}\big[-t\,c_i^\dagger c_j+\Delta_{ij}\,c_i^\dagger c_j^\dagger\big]\nonumber\\
&+\sum_{\langle\langle ij\rangle\rangle}\sigma^z_{l_1}\sigma^z_{l_2}\,\Delta_d\,c_i^\dagger c_j^\dagger+\text{h.c.},
\label{eq:merged}
\end{align}
with $B_p=\prod_{l\in p}\sigma^z_l$, $\Delta_{ij}=\Delta_0$ on $x$ bonds and $i\Delta_0$ on $y$ bonds, the diagonal pairing dressed by the two links of an elementary path, and $t=\Delta_0=\Delta_d=1$, $\mu=0$, $K=2$ (the companion paper's $C=-1$ parameters; at $K>0$ the plaquette term favours $\pi$ flux, the companion paper's static background, and we verify below that the all-$\pi$ sectors are the lowest ones). The physical space is the Gauss-law sector $G_v=(\prod_{l\ni v}\sigma^x_l)(1-2n_v)=+1$, of even total fermion parity; we work in a spanning-tree gauge in which the physical space is spanned by the $E-(V-1)$ non-tree links and $V-1$ fermions, $2L_xL_y$ qubits in all ($16$ at $4\times2$, $24$ at $4\times3$ and $6\times2$), the tree-link electric operators becoming products of $\sigma^x$ on the cut links and of the fermion parity of the subtree. The representation was certified against the full Hilbert space on $2\times2$ and $3\times2$ tori ($2\times10^{-14}$) and against a change of tree on the $4\times2$ Klein bottle ($3\times10^{-14}$). This is not the companion paper's model with a dynamical gauge field added: gauge covariance requires the same $\sigma^z$ on hopping and pairing, whereas the static model threads the $\pi$ flux through the hopping alone; the two are cousins, and none of the matter spectra of the four sectors below reproduces the static spectrum literally.

\paragraph{The doublet at $h=0$.} At $h=0$ the link configuration is conserved and the physical space decomposes into $2^{N_T}$ flux sectors ($N_T=E-V+1$ non-tree links: $512$, $8192$ and $8192$ sectors), in each of which the matter problem is quadratic. We solve all of them: the even-parity ground energy of a sector follows from its Bogoliubov--de Gennes spectrum and the parity of the Bogoliubov vacuum (the sign of a Pfaffian), a rule we validate against dense diagonalization of the $2^{V-1}$-dimensional even-parity block in every sector at $4\times2$ and in the lowest and in random sectors at $24$ qubits (agreement $5\times10^{-14}$). The four all-$\pi$ sectors are the four lowest; the lowest sector with violated plaquettes---necessarily a pair, since $\prod_pB_p=1$---lies $\delta_v=6.93$, $7.31$ and $6.87$ above the ground pair at $4\times2$, $4\times3$ and $6\times2$ (Table~\ref{tab:matter}), so that the all-$\pi$ restriction of the perturbation theory below is a statement about the spectrum, not an assumption. The four all-$\pi$ sectors are the holonomy sectors of the Klein bottle, $|H_1(K;\Z_2)|=4$, labelled by the Wilson loop $W_{\rm row}$ along the two-sided cycle and a Wilson loop $W_{\rm col}$ through the seam. The matter splits them $2+2$: the ground pair has $W_{\rm row}=-1$ (antiperiodic along the two-sided cycle) and $W_{\rm col}=\pm1$, and is exactly degenerate because the half translation $x\mapsto x+L_x/2$ maps one member to the other. It is the only lattice symmetry of Eq.~\eqref{eq:merged} on the Klein bottle besides fermion parity within the family we searched exhaustively: the $4L_xL_y$ affine maps $(x,y)\mapsto(s_xx+b,\,s_yy+c_y)$ of the quotient, which contain its $8L_y$ isometries (Sec.~\ref{sec:codes}; all $8L_y$ listed candidates are automorphisms of the gauge complex), each combined with a fermion $U(1)$ phase (multiples of $\pi/2$), an optional sublattice sign, complex conjugation and particle--hole. Only the identity, fermion parity, the half translation and their products survive; the row translations and the reflections, isometries of the lattice, are not symmetries of Eq.~\eqref{eq:merged}, because they move the seam, across which the chirality of the pairing is reversed. The upper pair lies $\delta=0.2189$, $0.6295$ and $0.0196$ above at $4\times2$, $4\times3$ and $6\times2$, with matter gaps $0.029$, $0.347$ and $0.245$ inside the ground sectors (Table~\ref{tab:matter}); the antiperiodic selection is the gauge-theory form of the companion paper's antiperiodic seam sector, and its weakening at $L_x=6$ is consistent with that paper's $L_x\bmod4$ rule, which moves the near-zero seam pair into the antiperiodic sector for $L_x\equiv2\pmod4$. On the doublet the seam-crossing Wilson loop $W_{\rm col}$ is diagonal and the half translation $T_{L_x/2}$ exchanges the two states, so $\bar Z=W_{\rm col}$, $\bar X=T_{L_x/2}$: in the language of Sec.~\ref{sec:codes}, the matter selects a two-dimensional subspace of the four-dimensional Klein-bottle code space, and the Dehn twist, which acts on the full code space as the $\mathrm{SWAP}$ of the two cross-cap qubits, acts on the selected pair as the logical $X$.

\begin{table}[!t]
\caption{The holonomy doublet of the merged model, Eq.~\eqref{eq:merged}. For each lattice: $E_0$, ground energy of the doublet at $h=0$; $\delta$, energy of the upper pair above it; $\Delta_m$, matter gap inside a ground sector; $\delta_v$, energy of the lowest sector with violated plaquettes above the doublet (all $2^{N_T}$ flux sectors enumerated). Splitting $\Delta(h)$ of the doublet under an electric field on the $y$ links ($y$), on the $x$ links ($x$) or on all links ($xy$): order $n$ and coefficient $2|m|$ of the exact leading-order perturbation theory, $\Delta_{\rm pert}=2|m|h^n$; the exact-diagonalization value $\Delta/h^n$ at the smallest field computed; and the local power $d\ln\Delta/d\ln h$ from the smallest to the largest field of the ranges listed in Appendix~\ref{app:methods}, first and last value.}
\label{tab:matter}
\centering\scriptsize\setlength{\tabcolsep}{3pt}
\begin{tabular}{llllll}
\toprule
size & field & $n$ & $2|m|$ & $\Delta/h^n$ (ED) & local power\\
\midrule
\multicolumn{6}{l}{$4\times2$: $E_0=-26.3370$, $\delta=0.219$, $\Delta_m=0.029$, $\delta_v=6.929$}\\
 & $y$ & $4$ & $3.387\times10^{-3}$ & $3.385\times10^{-3}$ & $4.00\to3.98$\\
 & $x$ & $4$ & $0.6533$ & $0.6534$ & $4.00\to3.95$\\
 & $xy$ & $4$ & $0.6500$ & $0.6500$ & $4.00\to3.79$\\
\multicolumn{6}{l}{$4\times3$: $E_0=-39.9043$, $\delta=0.629$, $\Delta_m=0.347$, $\delta_v=7.309$}\\
 & $y$ & $4$ & $2.090\times10^{-2}$ & --- & ---\\
 & $x$ & $6$ & $3.276\times10^{-3}$ & $3.368\times10^{-3}$ & $6.08\to6.17$\\
 & $xy$ & $4$ & $2.090\times10^{-2}$ & $2.099\times10^{-2}$ & $4.02\to4.47$\\
\multicolumn{6}{l}{$6\times2$: $E_0=-39.6434$, $\delta=0.020$, $\Delta_m=0.245$, $\delta_v=6.867$}\\
 & $y$ & $6$ & $3.432\times10^{-3}$ & $3.410\times10^{-3}$ & $5.98\to5.89$\\
 & $x$ & $4$ & $0.7095$ & $0.7096$ & $3.99\to3.97$\\
 & $xy$ & $4$ & $0.7095$ & --- & ---\\
\bottomrule
\end{tabular}
\end{table}

\paragraph{Two tunnelling routes.} An electric field $-h\sum_l\sigma^x_l$ splits the doublet by tunnelling between the two holonomy sectors, and the leading order is the smallest number of link flips that maps one sector onto the other. On the Klein bottle there are two such sets. Flipping the $L_x$ $y$ links of one row gap---a dual loop along the two-sided cycle---changes $W_{\rm col}$ directly, at order $h^{L_x}$. Flipping the $L_y$ $x$ links of a mirror-fixed column ($x=0$ or $x=L_x/2$)---a dual loop along a one-sided cycle---changes $W_{\rm row}$ and lands in the upper pair; a second such loop through the other fixed column returns to $W_{\rm row}=-1$ with $W_{\rm col}$ changed, at order $h^{2L_y}$, with the small energy $\delta$ as one of the denominators. The isotropic field therefore gives $\Delta\propto h^{\min(L_x,2L_y)}$, which is $h^4$ at all three sizes; the two routes are separated by fields on one link type only. We evaluate the leading order exactly: with $H_0$ the $h=0$ Hamiltonian, $|0\rangle,|2\rangle$ the ground states of the two sectors and $E_0$ their energy, the effective coupling is
\begin{align}
\langle2|H_{\rm eff}|0\rangle&=m\,h^{n},\nonumber\\
m&=\sum_{S}\sum_{\pi\in\mathfrak S_n}\Big\langle2\Big|\nonumber\\
&\qquad\prod_{k=1}^{n}\Big[(-\sigma^x_{l_{\pi(k)}})\,\frac{1}{E_0-H_0}\Big]\Big|0\Big\rangle ,
\label{eq:pert}
\end{align}
summed over the connecting link sets $S$ of size $n$ and all orderings of their flips (the last resolvent omitted), the resolvents acting in the intermediate flux sectors on the full even-parity matter space (dimension $2^{V-1}$, dense); the splitting is $\Delta_{\rm pert}=2|m|h^n$, the diagonal terms being equal by the half-translation symmetry. The connecting sets are found by enumeration: at $4\times2$ the two row gaps and the pair of fixed columns; at $4\times3$ the three row gaps at order $4$ and the pair of fixed columns at order $6$; at $6\times2$ the pair of fixed columns at order $4$ and the two row gaps at order $6$. Exact diagonalization at $h>0$ works in the full physical space of $2L_xL_y$ qubits, all flux sectors included, and resolves the doublet by the half-translation symmetry (the two states are the ground states of the two symmetry sectors): ARPACK for $16$ qubits, and Lanczos with symmetry projection at every step for the $24$-qubit lattices, converged to a residual $r<10^{-5}$. The Ritz value is a rigorous upper bound on the ground energy; for the error below it we use the Kato--Temple estimate $r^2/g$, with $g$ the gap to the next level of the same symmetry sector, taken from a second, deflated Lanczos run at representative fields (Appendix~\ref{app:methods}): $g$ is the distance to the upper pair at $6\times2$ ($0.019$--$0.020$) and the matter gap at $4\times3$ ($0.35$), and the estimates are below $5\times10^{-9}$, more than two orders of magnitude below the smallest $24$-qubit splitting quoted ($6.6\times10^{-7}$, at $4\times3$ and $h=0.075$). They are rigorous bounds under the interval condition stated in Appendix~\ref{app:methods}, which the deflated spectra support but which we have not proved.

Table~\ref{tab:matter} and Fig.~\ref{fig:matter} compare the two. Wherever the field is small enough for the leading order to dominate, the exact coefficient is reproduced: to $0.1\%$ or better at $4\times2$ for all three fields, to $0.4\%$ at $4\times3$ for the isotropic field (the $h^4$ term is the two-sided loop, whose coefficient $2.090\times10^{-2}$ matches the value $2.099\times10^{-2}$ measured at $h=0.075$; the one-sided route enters at $h^6$ and accounts for the drift of the local power from $4.02$ to $4.47$), to $0.6\%$ at $6\times2$ for the $y$ field, where the power is $6=L_x$: $5.98,5.96,5.93,5.89$ between $h=0.3$ and $0.9$. The one-sided route shows its own order separately, $6=2L_y$ at $4\times3$ ($3.37\times10^{-3}$ against $3.28\times10^{-3}$ at $h=0.3$, local powers $6.08$ and $6.17$) and $4=2L_y$ at $6\times2$ ($0.7096$ against $0.7095$ at $h=0.05$, local powers $3.99$ and $3.97$). The $4\times2$ lattice, where both routes are of order four and the fixed-column route dominates by two orders of magnitude, is the case in which a naive reading of the isotropic power as $L_x$ would have been wrong; the anisotropic fields and the exact coefficients are what make the attribution unambiguous.

\begin{figure*}[!t]
\centering\includegraphics[width=0.98\textwidth]{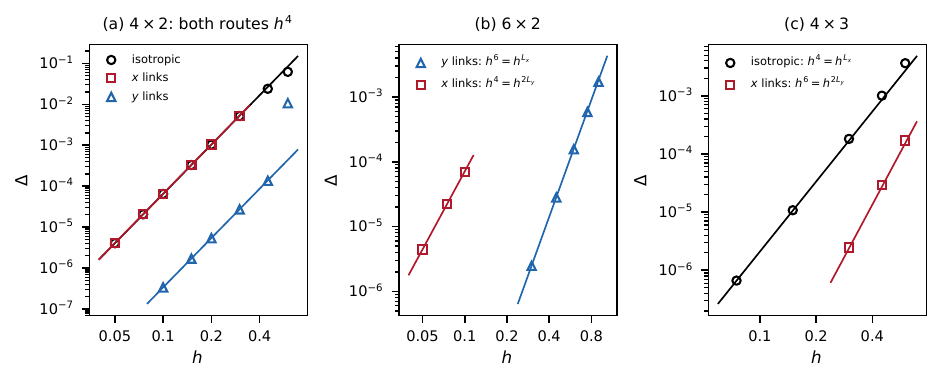}
\caption{Splitting of the holonomy doublet of the merged model against the field, exact diagonalization (symbols) and the exact leading-order perturbation theory $2|m|h^n$ of Eq.~\eqref{eq:pert} (lines). (a) $4\times2$: isotropic field and fields on the $x$ or $y$ links only; all three are $h^4$, the fixed-column route dominating. (b) $6\times2$: the $y$ field gives $h^{6}=h^{L_x}$, the $x$ field $h^{4}=h^{2L_y}$. (c) $4\times3$: isotropic field, $h^{4}=h^{L_x}$ with the $h^6$ correction of the fixed-column route, and the $x$ field, $h^{6}=h^{2L_y}$.}
\label{fig:matter}
\end{figure*}

\paragraph{Status.} On the lattices computed the doublet behaves as a holonomy qubit of the toric code: its splitting is of order $h^\ell$ with $\ell$ the length of the shortest dual loop that changes the holonomy, and the Klein bottle with chiral matter adds two features to that statement---the matter selects the antiperiodic pair, and the one-sided cycle, although the shorter of the two, enters only at twice its length because a single traversal leaves the selected pair. Its logical operators are a Wilson loop and a lattice translation; no other gate is available from Sec.~\ref{sec:codes} on the selected pair, since the colour-exchanging isometries are not symmetries of Eq.~\eqref{eq:merged}. We have reached three lattices with $L_x\le6$, $L_y\le3$; the selection energy $\delta$ is not monotone in size and is small at $6\times2$, the perturbative regime shrinks with $\delta$, and we do not know the fate of the doublet at large $L_x$ and $L_y$ together. The statements of this section are exact for the lattices computed and are offered as the gauge-theory counterpart of the companion paper's seam sector, not as a proposal for a protected qubit.

\section{Discussion}\label{sec:discussion}

\paragraph{Relation to twist defects and to parity-twisted topological order.} The dressed reflection gluing of the site register with $a+L_y$ even is a Klein bottle whose non-orientable cycle exchanges $e$ and $m$; its seam is a closed duality line of the kind whose open end points are the Ising-type twist defects of Refs.~\cite{Bombin2010,YouWen2012,BarkeshliJianQi2013,TeoRoyChen2014}. Being closed and end-point free, it binds no Majorana zero modes, and the code it defines is a stabilizer code with $\GSD=2$, the number of parity-invariant anyons \cite{ChanTeoRyu2016}. The edge register gives the untwisted Klein bottle, $\GSD=4$. Neither statement is new in substance; what the lattice comparison adds is that the two are the same lattice, the same surface and the same topological order, differing only in the register, and that a single colour-holonomy parity decides between them. The obstructed gluing is a third object: not a symmetry quotient at all, a local Hamiltonian in which the identification of the reversed corner pattern is made without the compensating duality, and whose seam is then forced to carry a pair of Majorana rings. We regard it as the stabilizer-model counterpart of the statement, proved for a chiral superconductor on the same real-space Klein bottle in the companion paper~\cite{Wang2026paper1}, that a gluing incompatible with the bulk symmetry action resolves itself in a localized seam structure rather than in an obstruction of the bulk.

\paragraph{The rings.} The exact solution rests on three facts: the anticommutation graph of the literal gluing is a union of cycles; a cycle of anticommuting involutions is a Majorana ring; and the bulk stabilizers lock the ring sectors. The resulting levels---$E_0$, the $(\mathrm P,\mathrm P)$ gap $4\tan(\pi/2L_x)$, the one-mode excitation $4\sin(\pi/L_x)$ (when the ring parities are not fixed by the bulk), the mixed sector at $2+2\tan(\pi/2L_x)$ for even $L_x$, and the one-mode gap $4\sin(\pi/2L_x)$ of the single $2L_x$-ring for odd $L_x$---are all reproduced on the lattice, as are the $24$-qubit energies. For the glide with odd $h$ the seam is four odd cycles, both product sectors of an odd cycle are degenerate, and the lowest excitation is a single violated bulk plaquette at energy $2$: the odd-cycle seam absorbs a single anyon, which on a torus could only be created in pairs. The rings are spectral objects of a frustrated Hamiltonian, with an exact but perturbatively fragile degeneracy; they are not the Majorana zero modes of twist-defect end points, and we do not identify them with any protected degree of freedom.

\paragraph{Relation to the chiral matter on the same surface.} The site register is the register of the $\pi$-flux fermions that motivate this work: on the same Klein bottle $(x,y+L_y)\sim(a-x,y)$, the chiral $p$-wave state of that model cannot be glued uniformly and forms a $\Delta C=2$ chirality wall along the seam, carrying two co-propagating Majorana branches whose finite-size levels select an antiperiodic-type sector depending on $L_x\bmod4$ (Ref.~\cite{Wang2026paper1}). The gauge-theory seam found here---two Majorana rings, one per column parity, for even $L_x$, or a single ring of twice the length for odd $L_x$, locked to the antiperiodic sector---sits at the same place and shares the antiperiodic structure, and for even $L_x$ the two-fold one. We record this as a parallel, not as an identification: the two objects are a chiral edge theory of a gapped two-dimensional superconductor and a frustrated stabilizer algebra, obtained from different Hamiltonians, and nothing in this paper relates them dynamically.

\paragraph{What is and is not established for the matter-coupled model.} Section~\ref{sec:matter} establishes, on three lattices, the ground-state structure of the gauged model in all flux sectors, the logical operator algebra of the selected doublet, and the order and the exact leading coefficient of its splitting under the electric field, with the two tunnelling routes identified and separated. It does not establish the fate of the doublet at large sizes (the selection energy $\delta$ is $0.63$ at $4\times3$ and $0.02$ at $6\times2$), a ledger of all local gauge-invariant perturbations (only the electric field is studied; matter perturbations that respect the gauge symmetry cannot change holonomies except through the same loops, but their coefficients are not computed), or anything about the relation between the selected antiperiodic pair and the seam levels of the companion paper beyond the parallel noted there. We also have not addressed interactions beyond the stabilizer terms, larger classes of gluings (crosscap insertions, gluings with both a reflection and a shift), or the $\Z_N$ generalization in which the odd-cycle seams would become parafermionic. On the code side we have computed exact distances and gates but no thresholds, decoders or noise models; the distances of Table~\ref{tab:codes} are those of a surface code of the same linear size, and we claim no advantage in overhead, only that the two stated topological actions are realized by the audited lattice permutations and Hadamard dressings in the abstract operation model.

\paragraph{Design lesson.} For code design the statement is simple: a non-orientable identification of a surface code produces a duality twist line or not depending on the register, and in the site register the parity of the axis offset plus the row number decides. A closed twist line costs nothing in stabilizer weight, halves the logical space and removes every non-Pauli gate in the audited isometry family; the untwisted Klein bottle keeps two logical qubits, one per cross-cap, on which the half translation is the $\mathrm{SWAP}$ and---in the site register only---a row shift with one row of Hadamards is $\Had\otimes\Had$. Both have depth one only in the permutation-as-one-layer convention stated above; no local-hardware depth claim follows. Together they realize the group that the mapping-class action and the anyon symmetry can give; a logical $S$ or a genuinely entangling gate between the cross-caps is not among them. An undressed identification is not a code but an exactly solvable frustrated model whose seam is a pair of Majorana rings. On the small lattices we could reach, chiral matter coupled to the gauge field selects a pair of holonomy sectors and turns the Dehn twist into the logical $X$ of a single doublet, whose splitting is that of the toric code with the one-sided loop counted twice; whether this survives at large sizes is open.

\section*{Acknowledgments}
The numerical work was designed and directed by the author. Code, figures and drafting were produced with AI assistance (Anthropic Claude); successive drafts of the manuscript were subjected to internal AI cross-review by separate OpenAI Codex instances, with the review reports and the author's point-by-point verifications retained in the original research archive (Appendix~\ref{app:code}); the sign frustration of the odd-$L_x$ dressed gluing and the parity condition on the single-ring-mode level (Secs.~\ref{sec:register} and \ref{sec:seam}) were found by an independent $\Z_N$ re-implementation of the models, calibrated against this paper's result files, and are corrected here. These computational and AI checks are internal reviews, not independent human peer review; this release does not claim independent human verification. The author takes full responsibility for the content. This work received no external funding.

\appendix

\section{Methods and sizes}\label{app:methods}
\emph{GF(2) tables.} Edge register: torus at $(L_x,L_y)=(4,4)$, $(6,4)$, $(4,6)$, $(6,6)$, $(8,4)$, $(4,5)$, $(6,5)$, $(5,4)$; mirror gluing with $a=0,1$ at the same sizes; sheared torus at the even-$L_x$ sizes. Site register: torus at $(4,4)$, $(6,4)$, $(4,6)$, $(6,6)$, $(5,4)$, $(4,5)$, $(5,5)$, $(8,4)$, $(7,4)$; sheared torus at $(4,4)$, $(6,4)$, $(4,5)$, $(6,5)$, $(8,4)$, $(8,5)$, $(10,4)$, $(4,6)$, $(6,6)$; mirror gluing, literal and dressed, with $a=0,1,2,3$ at $L_x\in\{4,6,8,10,12\}$ and $L_y\in\{3,4,5,6\}$ (thirteen sizes) and at the odd-$L_x$ sizes $(5,3)$, $(5,4)$, $(5,5)$, $(7,3)$, $(7,4)$, $(9,3)$, $(9,4)$, $(11,4)$ (twenty-one sizes in all); glide quotient, literal and dressed, at $(4,3)$, $(4,4)$, $(6,3)$, $(6,4)$, $(8,4)$, $(8,5)$, $(10,4)$, $(12,4)$. The rank is computed by Gaussian elimination over $\mathrm{GF}(2)$; anticommutation is the symplectic product of the rows; for every commuting set the phase of each relation (basis of the left null space, multiplied out with phases) is recorded, and the sign-frustrated odd-$L_x$ dressed gluing is diagonalized densely at $3\times3$, $3\times4$ ($a=0,1$) and $5\times2$ and by Lanczos at $5\times3$, with and without one plaquette sign reversed. \emph{Symmetry test.} On the $4\times4$ torus in the site register ($2^{16}$ states) the exact ground space is built by projecting random vectors with $\prod_p\tfrac12(1+F_p)$ and orthonormalizing; a candidate operation $U$ (a site permutation, optionally followed by $\Had^{\otimes N}$) is tested by $\max|UHv-HUv|/\max|Hv|$ on a random $v$, and its ground-space matrix $\langle g_i|U|g_j\rangle$ is diagonalized. \emph{Spectra.} Dense diagonalization for $n=12$; ARPACK Lanczos (real symmetric, $60$--$140$ Ritz values, $\mathrm{ncv}$ up to $290$) for $n=15,16,18,20,21$; a three-vector Lanczos iteration without reorthogonalization for the two $24$-qubit ground energies (converged to $5\times10^{-14}$ after $37$ and $32$ iterations). \emph{Central elements.} Ring products and alternating products are formed as Pauli strings with phases and evaluated on the exact eigenvectors. \emph{Degeneracy.} The ground projector of a cycle of length $L_c$ in a given sector is expanded in the $2^{L_c-1}$ ordered products $M_T=\prod_{j\in T}O_j$, $T\subseteq\{1,\dots,L_c-1\}$, with $c_T=\Tr(M_T^\dagger P)/2^m$ evaluated in the Jordan--Wigner representation on $m=\lceil L_c/2\rceil$ qubits (Appendix~\ref{app:cycles}; the expansion reproduces $P$ to $10^{-10}$), and multiplied by the idempotent $\tfrac12(1+\Pi_c/\pi_c)$ of the sector. The same products of the physical seam plaquettes are formed as Pauli strings with their phases, multiplied over the cycles and over the sector combinations of least energy that have nonzero trace, and every string is tested for membership in the $\mathrm{GF}(2)$ span of the bulk plaquettes; a member contributes $2^{n-N_b}$ times its coefficient times the phase of the matching product. The largest expansion has $2^{23}$ strings ($10\times4$ glide). Version 0.3 of the code built the cycle projector instead as the Lagrange polynomial in $H_c$ that vanishes on the other eigenvalues; that expansion is numerically unstable for long cycles (it returned $130.4$ for the $14$-cycle of the $7\times3$ mirror gluing, whose degeneracy is $128$) and is kept only as a cross-check for cycles of length $\le10$, where the two methods agree to $10^{-5}$. \emph{Perturbations.} Dense diagonalization at $4\times3$ with single-site $X$, $Y$ or $Z$ fields of strength $h\le0.08$, uniform or random, added to the seam rows $y=0,L_y-1$ or to the bulk row $y=1$. \emph{Codes (Sec.~\ref{sec:codes}).} Logical operators: a basis of the normaliser $N(S)$ from the $\mathrm{GF}(2)$ null space of the symplectic Gram matrix, reduced modulo $S$ and brought to symplectic form by Gram--Schmidt. Distances: for the site register all $2^{\,n-k}$ elements of the stabilizer group are enumerated as bit masks and the minimum weight of every logical coset is taken ($n\le24$); for the edge register, a CSS code, the $X$-type and $Z$-type parts are enumerated separately ($n\le48$); at $6\times5$ in the site register all $\sum_{w\le5}\binom{30}{w}3^w$ Pauli strings of weight $\le5$ are tested for commutation with the stabilizers and for non-membership in $S$. Automorphisms: each candidate operation (a qubit permutation induced by a lift of a lattice isometry, followed by $\Had$ on a set of rows) is applied to every stabilizer generator as a Pauli string with its phase, and the image is solved for as a product of generators; all $2^{L_y}$ unions of rows are tried for the site register and $\{\emptyset,\text{all}\}$ for the edge register. Logical actions are read off from the symplectic products of the conjugated logical generators with the basis, with the sign from the phases; the cross-cap basis is found by searching $\mathrm{Sp}(4,2)$ for a change of basis in which every automorphism acts as $1$, $\mathrm{SWAP}$, $\Had\otimes\Had$ or their product, with $\bar Z_i$ of $e$ type (the colour of the plaquettes violated by the cut loop). State-vector confirmation for $n\le20$: the exact code space from the stabilizer projector, the matrix of the operation on it, and the conjugated logical operators compared with the prediction. \emph{Merged model (Sec.~\ref{sec:matter}).} Lattices $4\times2$ ($16$ qubits), $4\times3$ and $6\times2$ ($24$ qubits), $K=2$, bond-centred axis. All $2^{N_T}$ flux sectors by enumeration of the non-tree link configurations: the matter energy of a sector from its $2V\times2V$ Bogoliubov--de Gennes matrix, $E=-\tfrac12\sum_k\epsilon_k$ plus the smallest $\epsilon_k$ when the Bogoliubov vacuum has odd parity (the sign of the Pfaffian of the Majorana form), validated against the dense even-parity block ($2^{V-1}$) in all $512$ sectors at $4\times2$ and in $64$ sectors (the $24$ lowest and $40$ random) at $4\times3$ and $6\times2$; the four all-$\pi$ sectors are also diagonalized densely. Symmetries by exhaustive search over the $4L_xL_y$ affine maps of the quotient (which contain the $8L_y$ isometries; the code checks that exactly $8L_y$ of them are automorphisms of the link--plaquette complex) combined with fermion $U(1)$ phases, sublattice signs, complex conjugation and particle--hole. Leading-order perturbation theory, Eq.~\eqref{eq:pert}: the connecting link sets are found by enumerating all subsets of the allowed links up to the order $\max(L_x,2L_y)$ and keeping those whose flips map sector $0$ to sector $2$; the sum over orderings is organised as a recursion over subsets, one dense solve of dimension $2^{V-1}$ per subset. Exact diagonalization: for $16$ qubits, ARPACK on the operator projected onto a half-translation sector (tolerance $10^{-13}$); for $24$ qubits, restarted block Lanczos (block $8$, full in-block reorthogonalisation, projection onto the symmetry sector at every step, locally optimal restart), converged to residual $r<10^{-5}$, warm-started from the matter ground state of a sector or from a neighbouring field. The gap $g$ to the next level of the same symmetry sector, needed for the bound $r^2/g$ on the ground energy, comes from a second Lanczos run deflated against the converged ground vector, started from the first excited matter state of the ground sector plus the ground states of both upper-pair sectors plus a random vector (so that every state of the sector has a nonzero overlap; the two upper-pair sectors are individually invariant under the half translation, with $T=+1$ and $-1$ respectively) and converged to residual $r_1\le10^{-3}$, with Rayleigh quotient $\theta_1$: $\theta_1-\theta_0=0.0196/0.0195$ ($6\times2$, $y$, $h=0.3$, $T=\pm1$), $0.0199/0.0198$ ($6\times2$, $x$, $h=0.1$), $0.3467/0.3467$ ($4\times3$, $x$, $h=0.3$). Kato--Temple, $\lambda_1\ge\theta_0-r^2/(\beta-\theta_0)$ for any $\beta\le\lambda_2$, is applied with the conservative choice $\beta=\theta_1-r_1$, which is a lower bound on $\lambda_2$ provided no eigenvalue of the sector lies in $(\theta_0,\theta_1-r_1)$; this interval condition is what the deflated run indicates (its start vector overlaps every state of the sector, and its Rayleigh quotient decreases monotonically to $\theta_1$) but is not a theorem. With it the error estimates are $4.5\times10^{-9}$, $1.1\times10^{-9}$, $3.0\times10^{-9}$, $1.3\times10^{-9}$, $2.8\times10^{-11}$, $2.7\times10^{-10}$ for the six ground energies (\texttt{results/matter.json} records $\theta_1$, $r_1$, $r$ and both the estimate with $\beta=\theta_1$ and the conservative one). Fields: $4\times2$, $h=0.05$--$0.6$ (isotropic and $x$), $0.1$--$0.6$ ($y$); $4\times3$, $h=0.075$--$0.6$ (isotropic) and $0.3$--$0.6$ ($x$); $6\times2$, $h=0.3$--$0.9$ ($y$) and $0.05$--$0.1$ ($x$). Local powers are logarithmic slopes between consecutive fields; for $4\times2$ they are quoted up to $h=0.45$, beyond which the leading order no longer dominates (at $h=0.6$ the $y$-field splitting jumps by two orders of magnitude, a level crossing inside the symmetry sector made possible by the small matter gap $0.029$; the point is shown in Fig.~\ref{fig:matter} but not used).

\section{Cycle algebras}\label{app:cycles}
In the Jordan--Wigner representation $\gamma_{2j}=X_jZ_{j-1}\cdots Z_0$, $\gamma_{2j+1}=Y_jZ_{j-1}\cdots Z_0$ on $\lceil L/2\rceil$ qubits, the operators $O_j=i\gamma_j\gamma_{j+1}$ ($j<L$) and $O_L=\pm i\gamma_L\gamma_1$ realize the cycle algebra in the two sectors of $\Pi$. Diagonalizing $-\sum_jO_j$: for $L=4,6,8,10,14$ the sector with the antiperiodic product has $E_0=-2.828427,-4,-5.226252,-6.472136,-8.987918$, unique, and the other sector has $E_0=-2,-3.464102,-4.828427,-6.155367,-8.762573$, two-fold degenerate, equal to Eq.~\eqref{eq:ring}; for $L=3,5,7$ both sectors have $E_0=-1.732051,-3.077684,-4.381286=-\cot(\pi/2L)$. The single-particle picture behind Eq.~\eqref{eq:ring}: $-\sum_ji\gamma_j\gamma_{j+1}$ has modes of energy $2\sin k$, occupied for $k\in(0,\pi)$ in the ground state, at $k=(2m+1)\pi/L$ (antiperiodic, $\sum_m\sin k=1/\sin(\pi/L)$) or $k=2\pi m/L$ (periodic, $\sum_m\sin k=\cot(\pi/L)$, with the zero modes at $k=0,\pi$).

\section{Code, data and reproducibility}\label{app:code}
The accompanying release is titled \emph{Paper 3 reproducibility package: An exactly solvable obstructed seam on a lattice Klein bottle}. Version 1.0.0 was published on Zenodo on 19 October 2026, DOI: \href{https://doi.org/10.5281/zenodo.23476997}{10.5281/zenodo.23476997}. Article text, original figures, package documentation and original result data use CC BY 4.0; original reproduction code uses the MIT License. Third-party dependencies retain their own licences.

The portable package preserves the original layout beneath \texttt{evidence/project/}. The directory \texttt{papers/paper3/code/} contains the register and Pauli algebra library, the register, seam, exact-degeneracy, perturbation, 24-qubit and logical-code scripts. Its \texttt{matter/} subpackage includes all engines, symmetry routines, Lanczos implementations, drivers and the summary reader required by the documented chain. Original saved results and historical run logs are preserved with source manifests. Internal review reports and manuscript development snapshots remain in the original research archive; they are not required inputs to the public package.

An explicit Python interpreter running \texttt{run.py quick} performs the documented representative replays, while \texttt{read} prints the preserved numbers. The \texttt{full} mode creates a fresh temporary working copy and executes both original stages without reusing saved outputs or checkpoints. Historical estimates are about two hours for the register and seam stage and about six hours for the matter stage, with about 5\,GB for each 24-qubit calculation. This release did not repeat those full chains. Fresh-extraction quick checks used CPython~3.13.5, NumPy~2.1.3, SciPy~1.15.3, Matplotlib~3.10.0, Numba~0.61.0 and llvmlite~0.44.0; the historical Python~3.11 environment is recorded separately. The package specifies dependency versions, cache isolation, numerical tolerances and the exact coverage of each mode.

The historical matter summary records source and result fingerprints. All 47 recorded raw-result hashes and all eight matter engine, driver and summarizer hashes match the preserved files. Three fingerprints for register or reporting files differ from the current tree, reflecting subsequent register and reporting revisions. The current matter engines do not import those three files; representative matter spectra and the complete saved-data summary were checked again. The unavailable earlier versions of those three files are not claimed to have been reconstructed. The package records this provenance difference separately from the numerical quick-check outcome. The original Lanczos and Kato--Temple qualifications in Appendix~\ref{app:methods} remain in force.

The source bundle includes the figures, local bibliography and matching \texttt{main.bbl}. The release changes the layout of the spectrum figure's legend only; its plotted arrays and axis ranges are checked against the original. No personal filesystem path is required for either the manuscript or the portable numerical package.

\setlength{\bibsep}{2pt}
\bibliographystyle{unsrtnat}
\bibliography{references}
\end{document}